\documentclass[conference]{IEEEtran}
\makeatletter
\def\ps@headings{%
\def\@oddhead{\mbox{}\scriptsize\rightmark \hfil \thepage}%
\def\@evenhead{\scriptsize\thepage \hfil \leftmark\mbox{}}%
\def\@oddfoot{}%
\def\@evenfoot{}}
\makeatother
\IEEEoverridecommandlockouts

\usepackage{cite}
\usepackage{graphicx}
\usepackage{textcomp}
\usepackage{xcolor}
\usepackage[T1]{fontenc}

\usepackage{amssymb}
\usepackage{enumerate}
\usepackage{adjustbox}
\usepackage{hyperref}
 \usepackage{url}
\usepackage[noend]{algpseudocode}
\usepackage{amsmath} 
\usepackage{amsthm}
\usepackage{tabularx}
\usepackage{tikz}
\usepackage{multirow}
\usepackage{threeparttable}
\usetikzlibrary{calc, positioning, arrows}
\usepackage[most]{tcolorbox}

\usepackage{enumitem}

\usepackage{esvect}
\usepackage{algorithmicx}
\usepackage[ruled,vlined]{algorithm2e}

\theoremstyle{definition} 
\newtheorem{definition}{Definition}[section] 
\newtheorem{theorem}{Theorem}[section]
\newtheorem{lemma}[theorem]{Lemma}

\newcommand{\remove}[1]{}

\usepackage{authblk}

\begin{document}

\title{Sponsored Group Signature and its Application to Privacy-preserving Guest Access in Smart Environments}


\author{Sepideh Avizheh, Reihaneh Safavi-Naini, Shiwei Sun} 
\affil{University of Calgary, Alberta, Canada}







\maketitle
\begin{abstract}
 Group signatures are privacy preserving signature schemes in which a group member can anonymously sign messages on behalf of the group, while providing accountability, by allowing  the signature of a misbehaving group member be ``opened''  and the  identity of the signer be revealed. In group signature members are admitted to the group by a (trusted) group manager. We motivate the need for a flexible mechanism in applications, such as privacy preserving access in smart environments,  and propose a two-level  member-join group signature that we call {\em SPonsored Group Signature (SPGS)} where group members of level 1 can ``sponsor’’ new members, in level 2,  to join the group.  This relaxation of user join comes with additional accountability mechanisms: we require that the signature of a sponsored member can be opened to the identity of the sponsor (that is sponsor is responsible for the sponsored member), and while all signatures are anonymous, for the sponsored members, the signatures are linkable. This allows a sponsor to efficiently identify an undesirable sponsored member. 
We formalize SPGS scheme,  define its security using a game-based approach, and give a generic construction of SPGS that uses a (dynamic) group signature scheme, a commitment scheme, and a knowledge-sound non-interactive zero knowledge proof of knowledge,  and prove its security.  We also give an instantiation of our construction. To show applicability of SPGS in practice, we consider the problem of providing guest access in a smart building, and introduce {\em Anonymous Guest Access Token  (AGAT)} that allows a temporary guest to anonymously access (a subset of) the building resources. We show how SPGS can be used (together with an IND-CPA secure public key encryption scheme) to give a direct construction for AGAT, and show the efficiency of our guest access protocol when it is instantiated with existing schemes.
\end{abstract}


.


\section{Introduction}
Group signatures were proposed by Chaum and van Heyst \cite{chaum1991group} as privacy preserving signature schemes that enable a group member to anonymously  sign messages on behalf of the group, while providing accountability by allowing the signature of a misbehaving group member to be ``opened'' and  the identity of the signer be revealed. 
The accountable privacy that is offered by group signatures has found 
many real-life applications in practice including   
corporate communication and remote attestation. 

In a basic  group signature, a  {\em Group Manager (GM)} sets up the parameters and the public  key, $gpk$, of  the system. A new group member interact with the group manager to obtain their  own private signing key $gsk$  that is used to sign messages. The signatures of  enrolled members can be    verified using $gpk$.   {\em Fully dynamic group signatures}  allow  users to  join at anytime during the lifetime of the system 
(they allow user revocation also). 
Group signatures provide {\em accountable signer anonymity} in the sense that signatures of two group members  with $gsk_1$ and $gsk_2$, are indistinguishable, 
but the GM (or a separate opening authority)  can {\em open} a signature and reveal the identity of the signer.

Two additional properties of anonymous signatures are  
{\em exculpability }
and {\em framing resistance} that aim to strengthen security of group signatures: 
exculpability protects users during the join phase and ensures that the GM cannot cheat users and give them a malformed credentials, and
framing resistance  protects them during the signing phase,  ensuring that the  GM, or others, cannot produce signatures that  implicates them.

An important property that 
seems to weaken the anonymity of group signatures, although it is  useful in many real-life scenarios, is {\em linkablity}
that maintains user anonymity while {\em links } together the signatures of a user.
Linkability,  is a very useful property in applications where  users' accesses to the system  must be monitored,  
and their mis-behavior detected.   Linkability can be used to anonymously profile users and ensures that users follow access policies of the system while maintaining their anonymity.
Signature opening  in a
linkable group signature \cite{zheng2018linkable}
effectively reveals the identity of the signer in all linked signatures.

{\em User join} 
in group signatures  needs 
interaction of the GM  with the   user and this necessitates  the GM to be always online and available.  {\em Ring signatures} are anonymous group signatures that remove the need for the GM: a signer can choose a group  of users (using public information of users) and sign a message as a member of the group, providing unconditional  (1 out $n$)  anonymity, with no possibility of  ``opening'' 
the signer of a message. This level of anonymity for signers is essential in applications such as whistle-blowing but it is limiting in accountable anonymous communication.

 Liu et al \cite{liu2004linkable} argued the need for {\em linkability  and ad hoc group formation}, that is removing trusted group manager   and  supporting some level of accountability, and introduced {\em Linkable Spontaneous Anonymous Group (LSAG) signatures}. LSAG is considered 
 as an extension of ring signature, where a user can form an {\em spontaneous } group of $n-1$ users   and sign as one of the $n$ users, while signatures of a user are linkable.
 
LSAGs, however, provide  (1 out $n$) anonymity without possibility of opening,   and  has signature length that is proportional to $n$.
 In comparison,  group signatures provide indistinguishability of the signers of any two signed messages (see Definition \ref{GSDefinition}), and can be constructed with constant length, but the join and open operations require a trusted GM.
 
\vspace{-5pt}
\subsection{Our work}
We propose an extension of a group signature that provides flexible  user join in the sense that enrolled users in the system can {\em sponsor} new users, hence allowing new users to join the system without interacting with the group manager, while maintaining accountability. Although relaxing the group join to be partially handled by the members could be a useful property, it  raises questions about security  of the system and accountability of the members that  join by other members. We show that linkablity can be used to provide the required accountability for this type of group membership.

\vspace{1mm}
{\em Our motivating scenario} (that is further discussed in Section \ref{Application})  is providing privacy preserving  guest access in smart buildings/homes.
To provide privacy preserving access for registered entities (aka, hosts/residents) of the building,  a group signature can be used where each host is a group member and uses their individual private key to anonymously authenticate themselves to the system and access the building's shared resources (e.g. printer). 
Our goal is to enable the host to provide {\em privacy preserving temporary access }{to the building's resources} to a guest, 
where privacy preserving  refers to anonymity of guests and hiding the host's identity. Ideally, we would like a guest who is ``approved'' by a host to be able to access the building resources without being linked to any particular host.
This scenario models many real-life settings including work environments that support visitors with temporary guest access such as universities,  research labs and corporations.  
Solutions  that 
use exiting cryptographic primitives (examples are given in Section \ref{Application}) result in costly and privacy invasive solutions. In particular, using existing group signatures require a GM that is always online, and that the need for the host interaction with the GM adds an additional complexity if guest's privacy is required.

\textbf{Sponsored group signature.} 
Our approach is to construct a  novel {\em  hierarchical two-level\footnote{Our approach can be generalized to multi-level signature systems,} group signature called Sponsored Group Signature (SPGS)} where users are in two distinct levels: the first (higher) level members, called {\em sponsors, $Sp$}
are enrolled in the system by interaction with the GM. Members of the second (lower) level, called {\em sponsored members, $Sm$,} consist of  members  who are enrolled in the system by the sponsors and without interaction with the GM.
The group public key, $gpk$, remains the same (as the group public key of sponsors) and the sponsored members do not need to interact with the group manager to get a new group public key. 
Intuitively, this structure solves the guest access problem as it allows a host (sponsor) to enroll a guest (sponsored member) in the system.  
SPGS systems, however, must ensure that the relaxation of join protocol to enroll sponsored members does not adversely affect security of the sponsor (original group member) signatures.  
In particular, a sponsor (is not assumed trusted --see security model of sponsored group signatures in Section \ref{SPGS}),  cannot abuse their user join capability, while anonymity and privacy of (honest) sponsored members are maintained. 

{\em Accountable flexible join.} We achieve the goal  of providing secure flexible join by requiring {\em linkable anonymity} for sponsored members, and {\em opening of  a sponsored-member's signature   to the corresponding sponsor.} 
This combination allows anonymous profiling of sponsored members who are not vetted by the group manager, and ensures  accountability of  sponsors  in using their additional new enrollment capability.

{\em 
Algorithms.}
 SPGS  consists of a tuple of  $\{Setup, Join^{attr}, Sign^{attr}, Verify^{attr}, Open^{attr}\}$ algorithms
 where $attr \in \{Sp, Sm\}$, defining 
 two sets of algorithms for the two levels of users,   $Sp$ and $Sm$.
In both levels, the private key of a member (sponsor or sponsored) is known to the member only, preventing framing of the user.
\noindent In Section \ref{SPGS}, we define the SPGS algorithms and their security properties using a game-based approach, following the framework of Bootle et al. \cite{bootle2016foundations} where adversary's access to the system is defined by a set of oracles. Definition \ref{SPGSDefinition} quantifies security properties of SPGS that are formalized in Fig. \ref{experiments}.
In Appendix \ref{GSModel} we also use the framework to define security of a  group signature that we use in our proofs inline with commonly used  definitions in \cite{bootle2016foundations}.

 \textbf{SPGS construction.} We propose a generic construction in the standard model which uses a partially dynamic group signature, a commitment scheme, and a knowledge-sound non-interactive zero knowledge proof of knowledge (NIZK) as building blocks (see Section \ref{SPGSconstruction}). Our SPGS construction is an extension of a group signature, $GS$, that is,
 the level 1 (sponsors) algorithms are the corresponding algorithms of the group signature. 
The level 2 (sponsored members) algorithms, consists of a group signature $GS$ and a Fiat-Shamir based NIZK.  When the sponsors join, they interact with the GM, at the end of which they get a group private and public key corresponding to the $GS$ scheme. When the sponsored-members join, they interact with the sponsor, choose a random key and use the commitment to the random key as their public keys. At the end of join phase, a sponsored-member gets a sponsor signature (which is a group signature) on their public key. The sponsored-member, when signs a message $m$ will generate a NIZK that, in the simplest form, shows that it knows the corresponding private key while considering the message to be signed as part of the public parameters of the NIZK.  
The NIZK proof and the sponsored-member signature can vary depending on the application context that SPGS is used; for example one can use a homomorphic commitment scheme to generate their public keys that allows the sponsored-member to re-randomize their keys and derive many keys (that can be used to register with different service providers). In this case, the NIZK proof additionally shows that the re-randomization has been done correctly. 
We prove the security of our construction and show it achieves the required security properties.
Additionally, we provide a concrete and efficient instantiation of the SPGS.

 \textbf{SPGS and its applications.}
We give an application of SPGS to the construction of {\em k-times Anonymous Guest Access Token} (k-AGAT) in a smart building that allows the guests use the building's resources. This has been one of the main motivation of proposing a new type of group signature. k-AGAT allows the host to determine its policy preferences such as expiry time when it issues a guest token. Also, it allows the verifier  to bound the number of guest tokens per each host to $k$ tokens. This is achieved since in k-AGAT the issued tokens by a host are linkable. Also, the tokens presented by the same guest are linkable and the guest's behavior can be traced. In case of misbehavior, the presented guest tokens can be opened to the identity of the host who issued the token. We define k-AGAT and its security requirements that are correctness, unforgeability, anonymity, and traceability in Section \ref{Application}. We also give a k-AGAT construction and show it satisfies the k-AGAT security and privacy requirements. In our k-AGAT construction, SPGS is used together with an IND-CPA secure public key encryption scheme, which allows the host act as the sponsor and signs the token content including the public key of the guest.  
The sponsored-member will transform the token using the sponsored-member signature of SPGS and present it to the verifier.  
We show how k-AGAT can be used  by a  guest to access building shared resources.

 \textbf{Implementation.} We implement cryptographic components  of the host, guest and IM for setup and guest access phase. Our implementation results are compared with a non-private solution that only uses a regular signature scheme rater than the group signature to issue guest tokens, show that the overhead our scheme proposes is tolerable while allowing privacy of hosts and guests. This confirms the feasibility and efficiency of our scheme.

\textbf{Paper organization.} Section \ref{RelatedWork} gives the related work. Section \ref{Prelim} gives the cryptographic primitives and preliminaries. Section \ref{SPGS} defines SPGS formally and give its security properties, Section \ref{SPGSconstruction} describes our SPGS generic construction and its security analysis. Section \ref{Application} provides the application of SPGS in constructing AGAT and secure guest accesses, 
and Section \ref{Conclusion} concludes the paper. 

\section{Related work}
\label{RelatedWork}
Group signatures were 
introduced by Chaum and Heyst \cite{chaum1991group}, allowing users to  anonymously signing on behalf of the group, and have found many extensions including \cite{buser2019dgm,liu2004linkable,bellare2005foundations,bootle2016foundations,chen1994new,gordon2010group,camenisch1997efficient,camenisch1998group,abhilash2022efficient,alamelou2017code,beullens2022group,fadavi2025dgmt}. 
SPGS  is a privacy preserving (two-level) hierarchical group signature  with two types of members (sponsor and sponsored), each belonging to one level of the hierarchy,  that  provides 
full anonymity for sponsors (level 1) and linkable anonymity for sponsored members (level 2), that supports  decentralized user join through sponsorship by members of the first level of the hierarchy (sponsors), and user accountability  by providing ``signature opening'' functionality by  a trusted group manager.
Related works can be broadly grouped into the following categories.

{\em Group signatures with delegation/hierarchies}.
Group signature is a widely studied privacy preserving signature scheme \cite{chaum1991group}. Partially and fully dynamic group signatures allow new users to join the group and/or existing  group members to be revoked. User join and revoke, however, are through a single group manager.
{\em Delegatable Anonymous Credentials (DAC)} \cite{belenkiy2009randomizable,belenkiy2008delegatable,mir2023practical,crites2019delegatable,blomer2018delegatable}
allow users to anonymously obtain credentials (signing rights) from authorities and delegate them to other users (also) prove possession of credentials).
Important properties of DAC are
{\em unlinkability} of  users interactions with the system,  support for multiple levels of 
{\em delegation}
and preserving
{\em anonymity of delegators}.  
The main differences  with SPGS are that in DACs, all participants  use the same cryptographic mechanisms and all delegated credentials are unlinkable.  In SPGS, however, the algorithms that are associated to  the sponsors and sponsored-members, including  
signing and verification algorithms, are different and and  signatures of a sponsored-member are linkable. In DACs \cite{mir2023practical, belenkiy2009randomizable, blomer2018delegatable}, addition  of a new member, however, is  through a single trusted group manager while in SPGS, it is decentralized.

 {\em Linkable Ring Signatures} (LRS) \cite{liu2004linkable,tsang2005short,liu2013linkable,tran2025many}, provide anonymity with {\em linkability} for signatures 
 and spontaneity  for group membership (no group secret or group manager). In an LRS scheme, a user can spontaneously form an ad-hoc group 
 and sign messages anonymously as a member of that group. 
Linkability in LRS is for  signatures of all users while in SPGS it is for sponsored members.

Other related schemes are {\em accountable tracing signatures} \cite{kohlweiss2014accountable} where the authority must prove it only opened specific users' signatures and did not abuse its power,
and {\em Bifurcated Anonymous Signatures (BiAS)} \cite{libert2021bifurcated} a user to determine conditions for tracing through a predicate and 
{\em Traceable Signature} \cite{kiayias2004traceable}
that refines the type of tracing.
These signatures, however, do not support decentralized member addition.

{\em Proxy Signatures} \cite{boldyreva2012secure}
allow delegation of signing rights to proxies that can be further delegated.

Finally an important property of SPGS is 
{\em  non-frameability and key independence} that ensures
private keys are only known to the users.
This property guarantees that  colluding parties cannot forge signatures that frame group members.

\section{Preliminaries} 
\label{Prelim}
In this section, we give the definition and notations used for the necessary cryptographic primitives (please see Appendix for the extended preliminaries covering the cryptographic primitives used in our application).

\textbf{Notation.} 
We use $out \leftarrow \mathcal{A}(in)$ to show that the algorithm  $\mathcal{A}$ is applied on the input $in$ to generate output $out$. We use $A: ops$ to Show party $A$ performs operations $ops$.

\subsection{Cryptographic primitives}

\noindent \textbf{Digital signature (DS)} 
consists of three algorithms: (i) $DS.KeyGen(1^\lambda)$,  which receives the security parameter $\lambda$ and outputs the private and public key pair $(Sk,Pk)$. (ii) $DS.Sign(Sk,\;m)$ which receives the private key $sk$ and a message $m$, and outputs a signature $\sigma$. (ii) $DS.Verify(\sigma,\;m,\;Pk)$ which takes the signature $\sigma$, the message $m$ and the public key $Pk$ and outputs $1$ if the signature is verified and $0$ otherwise. We consider the standard notion of existential unforgeability under chosen message attack $EUF\text{-}CMA$. 

\noindent \textbf{Group signatures} allow a member of a group (which is manged by the group manager $\mathcal{M}$)  to generate a signature anonymously such that a verifier can verify the validity of the signature knowing only the public parameters of the system. The manager is responsible for the correct functioning of the group, and can reveal the identity of the signer. We consider a partially dynamic group signature $GS$ that allows new members to join. $GS$ is composed of the following algorithms: (i) $GS.Setup(1^\lambda, setpp)=(pp, msk)$, is run by $\mathcal{M}$ and it takes the security parameter $\lambda$ and set up parameters $SP$, and outputs the public parameters $pp$ and the manager's secret key $msk$. (ii) $GS.Join(id_u,param_\mathcal{M})$ is an interactive algorithm that is run between a $\mathcal{M}$ and the user $id_u$ who desire to join the group (we assume all interactions takes place over a secure channel similar to \cite{buser2019dgm}). (iii) $GS.Sign(m, param_{u})= \sigma$, this algorithm is run by the member $id_u$ with its private parameters $param_{u}$, and outputs a valid anonymous signature $\sigma$. (iv) $GS.Verify(m,\sigma,pp)=0/1$, is run by the verifier and takes as input the message $m$, the signature $\sigma$, and the public parameters of the system $pp$, and outputs 1 if the signature is valid, or 0 otherwise. (v) $GS.Open(msk,\sigma)=id_u$, is run by $\mathcal{M}$ to reveal the identity of the signer of $\sigma$. 
We consider that $GS$ ensures the following security requirements: correctness, anonymity, non-frameability, and traceability.
For the security requirements we follow the definitions provided by Bootle et al. \cite{bootle2016foundations} and adapt them to a simpler setting were the group manager $\mathcal{M}$ act as both issuing authority and opening authority, and is assumed to be partially corrupted (i.e., its private state can be leaked).  
Also, we do not consider the $Judge$ algorithm 
since $\mathcal{M}$ follows the protocol correctly and traces the signatures if needed. Please see Appendix \ref{GSModel} for the security requirements and definitions.

\noindent \textbf{Public key encryption} 
consists of the following algorithms: (i) $E.KeyGen(1^\lambda)$ takes the security parameter as input and outputs the private/public key pair $(Sk,\;Pk)$. (ii) $E.Enc(Pk,\;m)$ takes the public key $Pk$ and message $m$ and outputs the ciphertext $C$. (iii) $E.Dec(Sk,\;C)$ takes the private key $Sk$ and ciphertext $C$ as input and outputs the message $m$. We consider an IND-CPA secure encryption scheme.

\noindent \textbf{Commitment} 
consists of the following algorithms. (i) $C.Setup(1^\lambda)$ receives the security parameter and outputs the commitment public parameters $cpar$, (ii) $C.Com(x)$ receives the input $x$ and outputs the opening information $d$ and the commitment $com$,  (iii) $C.Open(com,\; x,\; d)$ takes the commitment $com$, and the opening $x$ and $d$ as input and outputs $1$ if $x$ and $d$ are correct openings for $com$, and outputs $0$ otherwise. A commitment scheme ensures hiding and binding properties. %
In our concrete construction, we consider a homomorphic commitment scheme which includes the additional algorithm $HCom$ which allows to change the committed value by knowing the commitment $com$ on $x$. In another words, $C.HCom(com,\; x')$ outputs a commitment $com'$ on $x+x'$.  

\noindent \textbf{Non-interactive zero knowledge arguments of knowledge (NIZK)} 
consists of the following algorithms: (i) $NIZK.Setup(R)$ which takes the relationship $R$ (which implicitly defines the security parameter $\lambda$) and outputs the common reference string $crs$. 
(ii) $NIZK.Prove(crs, \rho, x, w)$ which takes the common reference string $crs$, the common input $\rho$, the statement $x$ that is being proven about $\rho$, and the witness $w$ as input and outputs the proof $\pi$.  
(iii) $NIZK.Verify(crs, \pi, \rho, x)$ receives the common reference string $crs$, the proof $\pi$, the common input $\rho$, and the statement being proven as input and outputs $1$ if the proof is verified and $0$ otherwise. We consider a NIZK that ensures perfect completeness, computational zero-knowledge, and computational knowledge-soundness properties (see Appendix \ref{NIZKModel} for definitions).

We consider sigma protocol for a relation $R$ that is a 3-move public coin protocol between a prover and a verifier both of which know the common input $\rho$ and the statement $x$, and the prover knows the witness $w$ such that $R(w,x)=1$. The prover sends the first message which is a commitment to a random value, and the verifier sends back a random challenge as the second message, and the the third message will be the prover's response on the first two messages and the witness $w$. Sigma protocols ensure zero knowledge (with honest verifiers) and knowledge soundness considering that there is an extractor who can extract a witness from the protocol transcripts if the protocol is run twice with the same first message and different challenges. The fiat-Shamir transform can be applied on any public coin protocol including a sigma protocol to obtain a non-interactive version of that. This transformation preserves the knowledge soundness property when it is analyzed in random oracle model.
We use sigma protocols in our concrete constructions.

\section{Sponsored group signature (SPGS)}
\label{SPGS}
 SPonsored Group Signature ($SPGS$) is a hierarchical group signature with two levels, supporting decentralized user join for members of the second level (sponsored members):
a sponsored-member joins ({\em spontaneously}) by a sponsor without the need to interact with the group manager.  SPGS offers full anonymity for sponsors  and linkable anonymity for sponsored members.

All signatures in SPGS can be opened to the sponsor's identity,   in the case of sponsored-members misbehavior. 
We define  SPGS with only user's join, and leave the extension of the scheme to provide efficient revocation for future work.

\textbf{Entities.} We consider the following entities:  {\em group manager,  sponsors (level 1 (L1) members), sponsored-members (level 2 (L2) members)}, and {\em verifiers}. 

\begin{enumerate}
    \item \textit{Group manager ($\mathcal{M}$)} is a 
     trusted entity that sets up the group and its parameters, and  generates the group public key $gpk \in pp$. To provide  accountability,  we enable the group manager to open the signatures (one can design a system such that this task is performed by a second authority).     
  Although we assume the group manager is trusted, in our
    security evaluation, we allow the group manager's private state 
    be leaked to the adversary.
    \item \textit{Sponsor ($Sp$)} is a L1 member in the group that is added to the group by $\mathcal{M}$), and obtains their private key $gsk \in param_{sp}$ through interaction with $\mathcal{M}$.
   A sponsor has the ability 
   to add L2  members  to the group and provide 
   them with the system's public parameters $pp$. 
    \item \textit{Sponsored-member ($Sm$)} is a L2 member of the group
    that is added to the group by an sponsor $Sp$, and receives their
    private key $sk \in param_{Sm}$, through interaction with the sponsor.
    \\
    Trust relationship  between $Sp$ and $Sm$ is delicate: they do trust each other in performing the cryptographic operations correctly, and  $Sp$ is willing to take the responsibility  (hence called sponsor) for
     $Sm$'s actions in the group in the sense that an $Sm$'s misuse of its signing capability can be traced to $Sp$.
    \item \textit{Verifiers ($V$)} can verify signatures of $Sp$s and $Sm$s using the group public key and public parameters.
\end{enumerate}

\noindent \textbf{Definition.} $SPGS$ is defined by a parameterized tuple of algorithms $\{Setup, Join^{attr}, Sign^{attr}, Verify^{attr}, Open^{attr}\}$
where the parameter $attr \in \{Sp,Sm\}$ specifies the associated type (level) of the member. 
The relationship between entities in terms of enabling group membership can be summarized as 
$\mathcal{M} \rightarrow Sp \rightarrow Sm$, denoting 
$\mathcal{M}$ enrolling $Sp$, and $Sp$ enrolling $Sm$.

\begin{itemize}
\item $Setup(1^\lambda, setpp)=(pp, msk)$, is run by $\mathcal{M}$ and it takes the security parameter $\lambda$ and set up parameters $setpp$, and outputs the public parameters $pp$ and the manager's secret key $msk$. 
\item $Join^{attr}( id_{attr}, param)$ is an interactive algorithm that is run either between a $\mathcal{M}$ and the user $id_{Sp}$ who wants to join as a sponsor if $attr=Sp$, or it is run between the sponsor and the user $id_{Sm}$ who wants to join as a sponsored-member if $attr=Sm$. It takes the private parameters $param$ of the enrolling entity belongs to $\{\mathcal{M}, Sp\}$, and outputs private parameters $param_{attr}$ to  the joining entity $id_{attr}$ and $out$ to the enrolling entity  belongs to $\{\mathcal{M}, Sp\}$ (we assume all interactions takes place over a secure channel similar to \cite{buser2019dgm}). $out$ shows the termination, whether the algorithm succeeds $\top$ or fails $\perp$.
\item $Sign^{attr}( m, param_{attr})= \sigma$, this algorithm is run by the sponsor or sponsored-member  depending on $attr \in \{Sp,Sm\}$, with its private parameters $param_{attr}$, and outputs a valid anonymous signature $\sigma$.
\item $Verify^{attr}( m,\sigma,pp)=0/1$, is run by the verifier and takes as input the attribute $attr \in \{Sp,Sm\}$,  the message $m$,  the signature $\sigma$, and the public parameters $pp$, and outputs 1 if the signature is valid, or 0 otherwise.
\item $Open^{attr}(msk, \sigma)=id_{Sp}$, is run by $\mathcal{M}$, it takes a sponsor or sponsored-member signature $\sigma$, where $attr \in \{Sp,Sm\}$, and reveals the identity of the sponsor of $\sigma$. 
    
\end{itemize}

\noindent \textbf{Security requirements.} $SPGS$ satisfies the following security requirements: {\em correctness}, {\em non-frameability}, {\em sponsor anonymity}, {\em sponsor traceability}, {\em sponsored-member privacy} and {\em sponsored-member linkability}. 
Our definitional framework is based on  the framework of Bootle et al. group signatures \cite{bootle2016foundations} and is  provided for the case of partially dynamic groups (group with join) and allows the group manager to be partially corrupted and its state leak. 
For sponsored-member linkability, our definition is inspired by the signer linkability notion defined in \cite{liu2004linkable}.

\begin{itemize}
  \item \textbf{Correctness} ensures that if an honest member joins as a sponsor or sponsored-member and generates a signature, their signature will be  verified.
  In our setting where the members can join at any time, these two conditions should hold for all honest members under any schedule under which the members join the group. Therefore, to formalize correctness we consider an adversary who can control users' joining  process and chooses the messages and the identity of the signers (both sponsor and sponsored-members). We require that for any messages and identities that the adversary chooses, the signatures generated can be verified correctly.
   \item \textbf{Non-frameability} captures unforgeability in a strong
sense and ensures the following:  
the collusion of all members except one honest group member $P$,   that has access to the private state of the group manager,  cannot generate a valid signature that, (i) if $P$ is of type sponsor (i.e.,$attr=Sp$), be opened to $P$, and (ii) if $P$ is of type sponsored-member (i.e.,$attr=Sm$), and  is added to the group through a correct execution of $Join^{Sm}$ by a sponsor $P'$,  be opened to $P'$.
We allow the adversary to have oracle access to the  member signing algorithm.
  \item \textbf{Sponsor anonymity} states that no one (including other sponsors and the verifier) can distinguish the sponsor who has generated a given signature from $id_{Sp_0}$ and $id_{Sp_1}$ better than a random guess even if they have seen the output of the $Open$ algorithm for some signatures. $\mathcal{M}$ is considered trusted for anonymity as $\mathcal{M}$ can always open the signatures and learn the identity of the sponsors. This definition is aligned with the signer anonymity definition of group signatures and captures sponsor signatures unlinkability as well.
  \item \textbf{Sponsor traceability} protects the group manager by ensuring that any valid sponsor or sponsored-member signature will be opened to the identity of a sponsor. This is achieved even if 
  the group manager's private state is leaked (i.e., $\mathcal{M}$ is considered partially corrupted). 
  \item \textbf{Sponsored-member privacy} is defined similar to sponsor anonymity in the sense that no one can distinguish the sponsored-member who has generated a given signature from $id_{Sm_0}$ and $id_{Sm_1}$ better than a random guess. However, for sponsored-member privacy, we exclude unlinkability from the definition by restricting the oracle access of the adversary to the signing oracle, that is, 
  the adversary cannot issue queries to the $SignHU(Sm, \cdot,\cdot)$ oracle for the challenged sponsored-members.   
  \item \textbf{Sponsored-member linkability} requires that the signatures of a sponsored-member be linkable. 
Our definition adapts the  signer linkability definition in \cite{liu2004linkable} to SPGS. We give
more details on this adaptation in Lemma \ref{LemmaLinkiability}.
\end{itemize}

\textbf{Oracles.} To formally define the security requirements, we use the following oracles (given in Figure \ref{oracles}).

\begin{itemize}
    \item $\mathbf{AddHU(attr, id_{attr})}$: This oracle allows the adversary to add honest users through honest execution of $Join^{attr}$ (to add sponsors or sponsored-members $attr \in\{Sp, Sm\}$) without learning their private parameters.
    \item $\mathbf{CorrU(id_{attr})}$ allows the adversary to corrupt the users (i.e., sponsors and sponsored-members) and learn  both their communication transcript when they run $Join^{attr}$ algorithm honestly and their private parameters (including signing keys). Thus, this oracle allows full exposure of all communication and keys, and can be run only immediately after $AddHU(attr, id_{attr})$.
    \item $\mathbf{AddCU(attr, id_{attr})}$: This oracle allows the adversary to add corrupted users  (sponsors or sponsored-members $attr \in\{Sp, Sm\}$) to the group. The adversary can deviate from the $Join^{attr}$ protocol and send arbitrary messages to the honest enrolling entity $e \in \{\mathcal{M}, Sp\}$ and see its output $out_e$. 
    \item $\mathbf{AChal_b(pp, m, id_{Sp_0}, id_{Sp_1})}$: This is a left-right oracle for defining sponsor anonymity. It takes as input the group public parameters $pp$, a message $m$, and two honest sponsors $id_{Sp_0}$ and $id_{Sp_1}$, and returns a sponsor signature
    on the message $m$ using the private parameters of $param_{Sp_b}$ for $b \leftarrow \{0,1\}$. The adversary can call this oracle once.
    \item $\mathbf{PChal_b(pp, m, id_{Sp}, id_{Sm_0}, id_{Sm_1})}$: This is a left-right oracle for defining sponsored-member privacy. It takes as input the group public parameters $pp$, a message $m$, the identity of an honest sponsor $id_{Sp}$, and two honest sponsored-members $id_{Sm_0}$ and $id_{Sm_1}$, runs $Join^{Sm}$ internally to obtain their private parameters, and returns a sponsored-member signature on the message $m$ using the private parameters of $param_{Sm_b}$ for $b \leftarrow \{0,1\}$. The adversary can call this oracle once.
    \item $\mathbf{LChal_b(Sm, pp, m_0, m_1, set_0, set_1)}$: This is a left-right oracle for defining sponsored-member linkability which is defined fpr $Sm$. It takes as input the  group public parameters $pp$, two messages $m_0$ and $m_1$, and two sets $set_0:\{id_{Sm_0}, id_{sm_0}\}$ and $set_1:\{id_{Sm_0}, id_{sm_1}\}$ and returns two group signatures $\sigma_0$ and $\sigma_1$ on message $m_0$ and $m_1$ respectively using the private parameters of the entities in $set_b$, for $b \leftarrow \{0,1\}$. The adversary can call this oracle once.
    \item $\mathbf{SignHU(attr,m,param_{attr})}$ This oracle is used by the adversary to obtain signatures for an honest user whose private parameters are not known by the adversary. It returns a signature $\sigma$ on the message $m$ using the $Sign^{attr}$ algorithm with the private parameters of $id_{attr}$, where $attr \in \{Sp, Sm\}$.
    \item $\mathbf{Open(attr, msk, m, \sigma, pp)}$ returns the identity of the sponsor $id_{Sp}$ who has generated or sponsored the signature $\sigma$. The oracle cannot be called on a signature obtained from the $AChal_b$, $LChal_b$, and $PChal_b$ oracles.
    \item $\mathbf{CorrM()}$ returns the private parameters of the manager $\mathcal{M}$.
\end{itemize}

\begin{figure*}[!bht]

\fbox{
    \begin{minipage}{0.48 \textwidth}
    \begin{small}
    $\mathbf{AddHU(attr, id_{attr})}$
    \hrule
    \begin{itemize}
        \item If $N>N_{max}$ Return $\perp$
        \item Update $\mathcal{H} \cup \{id_{attr}\}$ and $N=N+1$
        \item $param_{attr} \leftarrow Join^{attr}( id_{attr}, param)$
        \item Let the Join transcript is stored in $trscJoin_{attr}$
        \item Store $Q_{HU} =Q_{HU} \cup (id_{attr},trscJoin_{attr}, param_{attr}, out_e)$
        \item Return $pp$
    \end{itemize}

   $\mathbf{CorrU(id_{attr})}$
    \hrule
    \begin{itemize}
        \item If $id_{attr} \notin \mathcal{H}$ Return $\perp$
        \item Update $\mathcal{C} \cup \{id_{attr}\}$ and $\mathcal{H}=\mathcal{H} \backslash \{id_{attr}\}$
        \item Retrieve $(id_{attr},trscJoin_{attr}, param_{attr}, out_e)$ from $Q_{HU}$
        \item Return $trscJoin_{id_{attr}}$, $param_{attr}$, and $out_e$
   \end{itemize}

    $\mathbf{AddCU(attr, id_{attr})}$
    \hrule
    \begin{itemize}
        \item If $id_{attr} \in \mathcal{H}$ Return $\perp$
        \item Update $\mathcal{C} \cup id_{attr}$
        \item $param_{attr} \leftarrow Join^{attr}(id_{attr}, param)$
        \item Return $param_{attr}$ and $out_e$
    \end{itemize}

    $\mathbf{AChal_b(pp, m, id_{Sp_0}, id_{Sp_1})}$
    \hrule
    \begin{itemize}
        \item If $id_{Sp_0}, id_{Sp_1}$ $\notin \mathcal{H}$ Return $\perp$
        \item $ \sigma_b \leftarrow Sign^{Sp}(m, param_{Sp_b})$
        \item If $Verify^{Sp}(m, \sigma_b, pp)=0$ Return $\perp$
        \item Update $Q_{AChal} \cup (Sp, m, \sigma_b)$
        \item Return $\sigma_b$
    \end{itemize}
    \end{small}
        
    \end{minipage}
    \vline
    \vspace{20pt}
     \begin{minipage}{0.48 \textwidth}
     \begin{small}
     $\mathbf{SignHU(attr,m,param_{attr})}$
     \hrule
     \begin{itemize}
        \item If $param_{attr}= \perp$ Return $\perp$
        \item $ \sigma \leftarrow Sign^{attr}(m, param_{attr})$ 
        \item Update $Q_{Sign} \cup (attr, m,\sigma)$
        \item Return $\sigma$
    \end{itemize}

    $\mathbf{Open(attr, msk, m, \sigma, pp)}$
     \hrule 
     \begin{itemize}
        \item If $(attr, m, \sigma) \in Q_{Sign} \cup Q_{AChal} \cup Q_{LChal} \cup Q_{PChal}$ Return $\perp$
        \item If $Verify^{attr}(m, \sigma, pp)=0$ Return $\perp$
        \item Return $Open(msk, \sigma)$
    \end{itemize}

     $\mathbf{CorrM()}$
     \hrule
     \begin{itemize}
        \item Return $param_{\mathcal{M}}$
    \end{itemize}

    $\mathbf{PChal_b(pp, m, id_{Sp}, id_{Sm_0}, id_{Sm_1})}$
    \hrule
    \begin{itemize}
        \item If $id_{Sp}$ $\notin \mathcal{H}$ $\lor~ id_{Sm_0}, id_{sm_1} \in Q_{Join}$ Return $\perp$
        \item $\forall i \in\{0,1\}$ run $param_{Sm_i} \leftarrow Join^{Sm}(id_{Sm_i}, param_{Sp})$
        \item $ \sigma'_b \leftarrow Sign^{Sm}(m, param_{Sm_b})$
        \item If $Verify^{Sm}(m, \sigma'_b, pp)=0$ Return $\perp$
        \item Update $Q_{PChal} \cup (Sm, m, \sigma'_b)$
        \item Return $\sigma'_b$
    \end{itemize}

     $\mathbf{LChal_b(Sm, pp, m_0, m_1, set_0, set_1)}$ 
     \hrule
     \begin{itemize}
     \item Parse $set_{0}:\{id_{sm_0}, id_{sm_0}\}$ and $set_{1}:\{id_{sm_0}, id_{sm_1}\}$
        \item If $id_{Sm_0}, id_{Sm_1}$ $\notin \mathcal{H}$ Return $\perp$
        \item  $\forall i \in\{0,1\}$ run $\sigma_i \leftarrow Sign^{Sm}(m_i, param_{set_{b,i}})$
        \item $\forall i \in\{0,1\}$ if $Verify^{Sm}(m_i, \sigma_i, pp)=0$ Return $\perp$
        \item Update $Q_{LChal} \cup (Sm, m_0, \sigma_0) \cup (Sm, m_1, \sigma_1)$
        \item Return $\sigma_0, \sigma_1$
    \end{itemize}
   \end{small} 
     \end{minipage}
}
     \caption{Oracles used in experiments}
     \label{oracles}
     \vspace{-1em}
\end{figure*}

\begin{definition}
\label{SPGSDefinition}
    For a security parameter $\lambda \in \mathrm{N}$ and a PPT adversary $\mathcal{A}$, we say that $SPGS$ provides:
    \begin{enumerate}
        \item \textbf{Correctness} if there exists a negligible function $\nu_1$ such that  
        $Adv^{Corr}_{SPGS,\mathcal{A}} (\lambda)=Pr[Exp^{Corr}_{SPGS, \mathcal{A}}(\lambda)=1] \geq 1-\nu_1(\lambda)$\\

        \item \textbf{Non-frameability} if there exists a negligible function $\nu_2$ such that 
        $Adv^{Non-Frame}_{SPGS,\mathcal{A}} (\lambda)=Pr[Exp^{Non-Frame}_{SPGS, \mathcal{A}}(\lambda)=1] \leq \nu_2(\lambda)$\\

        \item \textbf{Sponsor anonymity} if there exists a negligible function $\nu_3$ such that 
        $Adv^{SpAnon}_{SPGS,\mathcal{A}} (\lambda)=Pr[Exp^{SpAnon-b}_{SPGS, \mathcal{A}}(\lambda)=1] \leq \frac{1}{2}+\nu_3(\lambda)$\\

        \item \textbf{Sponsor traceability} if there exists a negligible function $\nu_4$ such that 
        $Adv^{SpTrace}_{SPGS,\mathcal{A}} (\lambda)=Pr[Exp^{SpTrace}_{SPGS, \mathcal{A}}(\lambda)=1] \leq \nu_4(\lambda)$\\

        \item \textbf{Sponsored-member privacy} if there exists a negligible function $\nu_3$ such that 
        $Adv^{SmPriv}_{SPGS,\mathcal{A}} (\lambda)=Pr[Exp^{SmPriv-b}_{SPGS, \mathcal{A}}(\lambda)=1] \leq \frac{1}{2}+\nu_5(\lambda)$\\

        \item \textbf{Sponsored-member linkability} if there exists a negligible function $\nu_5$ such that 
        $Adv^{SmLink}_{SPGS,\mathcal{A}} (\lambda)=Pr[Exp^{SmLink-b}_{SPGS, \mathcal{A}}(\lambda)=1] \geq 1-\nu_6(\lambda)$\\
    \end{enumerate}

    where $Exp^{Corr}_{SPGS, \mathcal{A}}$, $Exp^{Non-Frame}_{SPGS, \mathcal{A}}$, $Exp^{SpAnon}_{SPGS, \mathcal{A}}$, $Exp^{SpTrace}_{SPGS, \mathcal{A}}$, $Exp^{SmPriv}_{SPGS, \mathcal{A}}$, $Exp^{SmLink}_{SPGS, \mathcal{A}}$ are defined in Figure \ref{experiments}.
\end{definition}

\begin{figure}[!bht]

\fbox{
    \begin{minipage}{0.95\columnwidth}
    \begin{small}
    $Exp^{Corr}_{SPGS, \mathcal{A}}$
    \hrule 
    $(pp,msk) \leftarrow Setup(\lambda,setpp)$, $N=0$\\
    $(m, id_{Sp}, m', id_{Sm}) \leftarrow \mathcal{A}^{AddHU,CorrU,AddCU,CorrM}(pp)$\\
    If $id_{Sp} \notin \mathcal{H} \lor id_{Sm} \notin \mathcal{H}$ Return $0$\\
    $\sigma \leftarrow Sign^{Sp}(m,param_{Sp})$\\
    $\sigma' \leftarrow Sign^{Sm}(m', param_{Sm})$\\
    If $ Verify^{Sp}(m,\sigma,pp)=0$ Return $0$\\
    If $ Verify^{Sm}(m',\sigma',pp)=0$ Return $0$\\
    \textbf{Return} 1 \\

    $Exp^{Non-Frame}_{SPGS, \mathcal{A}}$
    \hrule 
    $(pp,msk) \leftarrow Setup(\lambda,setpp)$, $N=0$, $Q_{Sign}=\emptyset$\\
    $(attr, id_{Sp}, id_{Sm}\perp, m,\sigma_{attr}) \leftarrow \mathcal{A}^{AddHU,CorrM,SignHU}(pp)$\\
     If $attr=Sp \land id_{Sp} \notin \mathcal{H}$ Return $0$\\
     If $attr=Sm \land id_{Sm} \notin \mathcal{H}$ Return $0$\\
    If $Verify^{attr}(m,\sigma_{attr},pp)=0$ Return $0$\\
    If $(attr,m, \sigma_{attr}) \in Q_{Sign}$ Return $0$\\
    \textbf{Return} $Open^{attr}(msk, \sigma_{attr}) \overset{?}= id_{Sp}$\\

    $Exp^{SpAnon-b}_{SPGS, \mathcal{A}}$
    \hrule 
    $(pp,msk) \leftarrow Setup(\lambda,setpp)$, $N=0$, $Q_{AChal}=\emptyset$\\
     $b' \leftarrow \mathcal{A}^{AddHU,CorrU,AddCU,Open,AChal_b}(pp)$\\
     If $b' \neq b$ Return $0$\\
    \textbf{Return} 1\\ 

     $Exp^{SpTrace}_{SPGS, \mathcal{A}}$
    \hrule 
    $(pp,msk) \leftarrow Setup(\lambda,setpp)$, $N=0$\\
    $(attr,m,\sigma_{attr})  \leftarrow \mathcal{A}^{AddHU,CorrU,AddCU,Open,CorrM}(pp)$\\
     If $Verify^{attr}(m,\sigma_{attr},pp)=0$ Return $0$\\
    \textbf{Return} $Open^{attr}(msk, \sigma_{attr}) \overset{?}= \perp$\\

    $Exp^{SmPriv-b}_{SPGS, \mathcal{A}}$
    \hrule 
    $(pp,msk) \leftarrow Setup(\lambda,setpp)$, $N=0$, $Q_{PChal}=\emptyset$\\
     $b' \leftarrow \mathcal{A}^{AddHU,CorrU,AddCU,Open,PChal_b}(pp)$\\
     If $b' \neq b$ Return $0$\\
    \textbf{Return} 1\\ 

    $Exp^{SmLink-b}_{SPGS, \mathcal{A}}$
    \hrule 
    $(pp,msk) \leftarrow Setup(\lambda,setpp)$, $N=0$, $Q_{LChal}=\emptyset$\\
     $b' \leftarrow \mathcal{A}^{AddHU,CorrU,AddCU,Open,LChal_b}(pp)$\\
     If $b'\neq b$ Return $0$\\
     \textbf{Return} 1
    
    \end{small}
    \end{minipage}
    }
     \caption{Security games of $SPGS$}
     \label{experiments}
      \vspace{-1.5em}
\end{figure}

\textit{Linkability of the sponsored-member signature}. 
Our definition of linkability follows from the signer linkability notion defined in \cite{liu2004linkable} but revises the definition to an equivalent one, as outlined below.
Consider a PPT algorithm $F$ which takes two signatures and outputs 1 if the two signatures are linked (signed by the same user) and 0 otherwise. Liu et al. \cite{liu2004linkable} definition of linkability requires that the two following probability statements hold: (i) The probability that $F$ outputs 0 when the two signatures are generated by the same entity is negligible, i.e., $Pr[F(m_0, m_1, \sigma_0, \sigma_1)=0: id_{Sm_0} = id_{Sm_1}] \leq \nu(\lambda)$, and (ii) the probability that $F$ outputs 1 when the two signatures are generated by two distinct users is negligible, i.e., $Pr[F(m_0, m_1, \sigma_0, \sigma_1)=1: id_{Sm_0} \neq id_{Sm_1}] \leq \nu(\lambda)$, for any $id_{Sm_0}, id_{Sm_1}$ and any messages
$m_0$ and $m_1$, and any $\sigma_0 \leftarrow Sign(Sm, m_0, param_{Sm_i})$, $\sigma_1 \leftarrow Sign(Sm, m_1, param_{Sm_i})$ and  $i \in \{id_{Sm_0}, id_{Sm_1}\}$. 
We define sponsored-member linkability using the random variable $Exp^{SmLink-b}_{SPGS,\mathcal{A}}$, and the advantage of the PPT algorithm $\mathcal{A}$. The adversary $\mathcal{A}$ will choose a pair of identities, $id_{Sm_0}$ and $id_{Sm_1}$, and must decide if a pair of received signatures are signed by $id_{Sm_0}$, or two entities ($id_{Sm_0}$ and $id_{Sm_1}$) by outputting 1 or 0, to show whether the two signatures are from the same identity or not. 
We show that this is equivalent to the definition in \cite{liu2004linkable} in Lemma \ref{LemmaLinkiability} in Appendix \ref{Linkability}.

\subsection{Our construction}
\label{SPGSconstruction}
In the following, we give a generic construction using a (partially) dynamic group signature scheme denoted by $GS$ consisting of $(GS.Setup, GS.Join, GS.Sign, GS.Verify, GS.Open)$ algorithms, a  
commitment scheme denoted by $C$ consisting of  $(C.Setup, C.Commit, C.Open)$ algorithms, and a knowledge-sound non-interactive zero knowledge proof of knowledge denoted by $NIZK$ consisting of $(NIZK.Setup, NIZK.Prove, NIZK.Verify)$ algorithms (see Section \ref{Prelim} and Appendix \ref{GSModel} and \ref{NIZKModel} for their definitions and security requirements).

To join, $Sp$ interacts with $\mathcal{M}$ and calls $GS.Join$ to get the group private and public key. The signature of $Sp$ is a group signature which will be generated by running $GS.Sign$. To enroll $Sm$ to the system,  $Sm$ will generate a pair of private and public keys and will receive the signature of $Sp $ on its public key. The public key of $Sm$ is generated by running the commitment $C.Commit$ on a random vbalue chosen by $Sm$, which is used as its private key. $Sm$'s signature on message $m$ is obtained by generating a NIZK proof with the private key as the witness and the message $m$ as part of the public parameters. The signature of $Sm$ consists of the obtained sponsor signature $\sigma$ on its public key $Pk_{Sm}$, the public key, $Pk_{Sm}$, and the NIZK proof $\pi$. $(Pk_{Sm}, \pi)$ can be seen as a signature of knowledge (SOK) (see Appendix \ref{SOK} for SOK definition).  
Below we give the details of the scheme.

\begin{itemize}
\item \underline{$Setup(1^\lambda, setpp) \rightarrow (pp, msk)$}, is run by $\mathcal{M}$, it takes the security parameter $\lambda$, and setup parameters $setpp$ (including the relation $R$ for NIZK), and perform the followings:

\begin{itemize}
    \item runs $(pp', msk)=GS.Setup(1^\lambda,setpp)$
    \item runs $crs=NIZK.Setup(R)$
    \item runs $cpar=C.Setup(1^\lambda)$
    \item outputs the public parameters $pp= (pp',cpar, crs)$ and the manager's secret key $msk$. 

\end{itemize}
\item \underline{$Join^{Sp}(id_{Sp}, param_{\mathcal{M}}) \rightarrow param_{Sp}$} is run between a manager $\mathcal{M}$ and the user $id_{Sp}$ who desire to join the group and act as a sponsor.

\begin{description}
    \item[Round (1)] $Sp$ sends its identity  $id_{Sp}$ to the manager $\mathcal{M}$.
    \begin{itemize}
    \item $\mathcal{M}$ runs $param_{Sp}=GS.Join(id_{Sp})$.
    \end{itemize}
    \item[Round (2)] $\mathcal{M}$ sends  $param_{Sp}$ to $Sp$.
    \begin{itemize}
        \item $\mathcal{M}$ outputs $\top$.
    \end{itemize}
\end{description}
\item \underline{$Join^{Sm}(id_{Sm}, param_{Sp}) \rightarrow param_{Sm}$}, this interactive algorithm is run between  the sponsor  $id_{Sp}$ and the sponsored-member $id_{Sm}$.

\begin{description}
    \item[Round (1)] $Sp$ sends the commitment parameters  $cpar$ to $Sm$.
    \begin{itemize}
        \item $Sm$ generates a private key, denoted by $sk_{Sm}$, and computes $(Pk_{Sm}, r)=C.Com(sk_{Sm})$. 
    \end{itemize}
    \item[Round (2)] $Sm$ sends its public key $Pk_{Sm}$ to $Sp$.
    \begin{itemize}
        \item $Sp$ sets $m=Pk_{Sm}$, and runs $\sigma= GS.Sign(m, param_{Sp})$.
    \end{itemize}
    \item[Round (3)] $Sp$ sends the private parameters $param_{Sm}=(pp,\sigma,Pk_{Sm})$  to $Sm$.
    \begin{itemize}
        \item $Sp$ outputs $\top$.
    \end{itemize}
\end{description}

\item \underline{$Sign^{Sp}( m, param_{Sp}) \rightarrow \sigma$}, this algorithm runs by the sponsor $id_{Sp}$. It  takes $m$ and $param_{Sp}$ as input, and 
\begin{itemize}
   \item computes $\sigma= GS.Sign(m, param_{Sp})$
   \item outputs the sponsor signature $\sigma$.
\end{itemize}
\item \underline{$Sign^{Sm}(m, param_{Sm}) \rightarrow \sigma'$}, is run by the sponsored-member $id_{Sm}$. This algorithm 
\begin{itemize}
    \item parses $param_{Sm}=(pp,\sigma,Pk_{Sm})$
   \item  generates a zero-knowledge proof as $\pi=NIZK.Prove(crs, \rho,x, w)$ where $\rho= (m,Pk_{Sm})$, $w= (sk_{Sm}, r)$, and $x=[\exists (sk_{Sm}, r): (Pk_{Sm}, r)=C.Com(sk_{Sm})]$.
    \item outputs the sponsored-member signature  $\sigma'= (\sigma,Pk_{Sm},\pi)$.
\end{itemize}
\item \underline{$Verify^{Sp}(m,\sigma,pp)\rightarrow 0/1$}, is run by the verifier and takes as input the message $m$, the sponsor signature $\sigma$, and the public parameters of the system $pp$, it 
\begin{itemize}
  \item verifies the sponsor signature $b=GS.Verify(m, \sigma, pp)$
  \item outputs $b$.
\end{itemize}
\item \underline{$Verify^{Sm}(m,\sigma',pp) \rightarrow 0/1$}, is run by the verifier and takes as input the message $m$, the sponsored-member signature $\sigma'$, and the public parameters of the system $pp$, it 
\begin{itemize}
    \item parses the sponsored-member signature $\sigma'=(\sigma,Pk_{Sm},\pi)$
    \item verifies the sponsor signature  $b= GS.Verify(Pk_{Sm}, \sigma, pp)$
    \item verifies the zero-knowledge proof $b'= NIZK.Verify(crs, \pi, \rho, x)$ where $\rho= (m,Pk_{Sm})$ and  $x=[\exists (sk_{Sm}, r): (Pk_{Sm}, r)=C.Com(sk_{Sm})]$. 
    \item outputs 1 if both $b=1$ and $b'=1$, or 0 otherwise.
\end{itemize}
\item \underline{$Open^{SP}(msk, \sigma') \rightarrow id_{Sp}$}, is run by $\mathcal{M}$;  it computes $id_{Sp}=GS.Open(msk, \sigma)$, and outputs $id_{Sp}$.
\item \underline{$Open^{Sm}(msk, \sigma') \rightarrow id_{Sp}$}, is run by $\mathcal{M}$; it parses $\sigma'=(\sigma,Pk_{Sm},\pi)$, then computes $id_{Sp}= GS.Open(msk, \sigma)$, and outputs $id_{Sp}$. 
    
\end{itemize}

\vspace{-5pt}
\subsection{Security analysis} The following theorem summarizes our security analysis:

\begin{theorem}
    Our generic construction of $SPGS$ given in Section \ref{SPGSconstruction} achieves correctness, non-frameability, sponsor anonymity, sponsor traceability, and sponsored-member linkability assuming $GS$ ensures correctness, unforgeability, anonymity, traceability, the $NIZK$ ensures completeness, computational zero knowledge, and computational knowledge soundness, and $C$ is a hiding and binding additive homomorphic commitment scheme. 
    
\end{theorem}

\noindent \textbf{Proof.}
The proof of theorem consists of a set of lemmas, one for each property, that are outlined below. The detailed proof sketch using hybrid games is given in the Appendix \ref{Proofs}.

\begin{lemma}
    \textbf{Correctness.} Our generic construction of $SPGS$ achieves correctness,  if the group signature $GS$ and commitment $C$ ensure correctness and $NIZK$ ensure completeness.
\end{lemma}
The proof is straightforward.

\begin{lemma}
    \textbf{Non-frameability.} Our generic construction of $SPGS$ achieves non-frameability,  if $NIZK$ ensures knowledge soundness, the commitment $C$ is binding, and the group signature $GS$ is non-frameable.
\end{lemma}

We consider two cases, where (i) $attr=Sp$, and  (ii) $attr=Sm$. In (i) the adversary outputs a sponsor signature, which is a group signature. Due to the non-frameability of $GS$, the success probability of the adversary will be negligible.  In (ii) the adversary outputs a sponsored-member signature which consists of a group signature, a committed public key, and a NIZK proof. Since $GS$ is non-frameable, and the join algorithm is executed honestly, the success probability of adversary to output a valid $GS$ 
is negligible. Additionally, NIZK ensures knowledge soundness which means that the probability that adversary outputs a valid NIZK proof and the extractor extracts the witness $w$ such that $(w,x) \notin R$ is negligible. Furthermore, the commitment scheme is binding and the adversary cannot equivocate the private key $w$. Therefore, the success probability of adversary to output a valid NIZK proof for a public key that has been generated during the honest execution of Join (without knowing the witness) is negligible.

\begin{lemma}
    \textbf{Sponsor anonymity.} Our generic construction of $SPGS$ achieves sponsor anonymity,  if the group signature $GS$ ensures anonymity for the signer. 
\end{lemma}

For a sponsor or sponsored-signature the only part that is related to the sponsor is the group signature. The NIZK proof  and the sponsored-member public key are independent of the identity of the sponsor and they can be generated  using a random private key.  Due to the anonymity of $GS$, sponsor anonymity is obtained. 

\begin{lemma}
    \textbf{Sponsor traceability.} The construction of $SPGS$ given in Section \ref{SPGSconstruction} achieves sponsor traceability,  if the group signature $GS$ ensures traceability.
\end{lemma}

Sine the open algorithm uses the group signature (generated by the sponsor) to find the identity of the sponsor, the  sponsor traceability is achieved because of the traceability of $GS$.

\begin{lemma}
    \textbf{Sponsored-member privacy.} Our generic construction of $SPGS$ achieves sponsored-member privacy,  if $NIZK$ is a zero-knowledge argument and the commitment scheme $C$ is hiding. 
\end{lemma}

We note that our construction does not include any information about the identity of the signer except for the private and public key information that is generated by sponsored-member. The public key is shared with the sponsor, and since the commitment is hiding no information about the private key is revealed. Also, the NIZK is zero knowledge and it does not leak any information about the private key. So, the adversary cannot learn any information related to the sponsored-member and cannot identify the signer. 

\begin{lemma}
    \textbf{Sponsored-member linkability.} Our generic construction of $SPGS$ achieves sponsored-member linkability,  if 
    $GS$ ensures non-frameability.
\end{lemma}

All the signatures of a sponsored-member include the sponsor signature on their public keys that they have obtained during the Join algorithm. Since $GS$ is non-frameable, no one can forge a valid sponsor signature to break the linkability. Therefore, the adversary can always link the signatures of a sponsored-member with high probability.

\subsection{Concrete construction}
\label{SPGSInit}
To instantiate our construction, we use the BBS group signature scheme \cite{BBS} for $GS$. BBS is a static group signature scheme with known group size in the key generation phase. We set $n$ to be large enough to accommodate the maximum group size. 
For commitment scheme, we consider the following Pedersen commitment scheme (on EC) that allows re-randomization and it is as follows: $com=C.Com(x)=g^xh^r$, where $r=0$ for simplicity, and $C.HCom(com, r')= com. g^{r'}=g^{x+r'}=C.Com(x+r')$. 
For $NIZK$, we consider a sigma protocol-based NIZK (on EC) for discrete log using Fiat-Shamir transform given in the random oracle model in Algorithms \ref{ZK} and \ref{ZK2}. Note that our NIZK in the sponsored-member signature takes a message $m$ as input (which can be seen as a signature of knowledge) thus in the NIZK algorithms given in \ref{ZK} and \ref{ZK2}, we also include $m$ in the public parameters. $m$ is given as input to $\mathcal{H}$ when computing the challenge. We give the NIZK for the case that $Sm$ re-randomizes its key to get a new key after getting the sponsor signature. This allows $Sm$ to use different keys for different applications. Note that in this case, all the signatures of $Sm$ can still be linked to each other, since all of them have been originated from the same key. 
\vspace{-1em}

\begin{algorithm}
\small
\caption{NIZK.Prove}
\label{ZK}
\SetKwInOut{Input}{Input}
\SetKwInOut{Output}{Output}

 \Input{Witness $w=(sk_{Sm}, r, r')$ and statement $x= [\exists (sk_{Sm}, r, r'): Pk'_{Sm}=C.HCom(Pk_{Sm}, r') \land sk'_{Sm}=sk_{Sm}+r' \land (Pk_{Sm}, r)=C.Com(sk_{Sm})] $, the common input $\rho=(Pk_{Sm},Pk'_{Sm}, m)$, where $m$ is an arbitrary message, and the hash function $\mathcal{H}$}
 \Output{The proof $\pi$}
 \begin{algorithmic}[1]
     \State  Choose two random values $a$ and $b$ in $Z^*_p$
     \State Compute the commitment $A=g^a \in G$ and $B=g^b \in G$
     \State Compute the challenge $c=\mathcal{H}(A,B,Pk_{Sm}, Pk'_{Sm},m)$ 
     \State Compute the response $d_1=a+c \times r'$ and $d_2= b+ c \times sk_{Sm}$\\
    \Return $\pi=(A,B,d_1, d_2)$    
 \end{algorithmic}   
\end{algorithm}
\setlength{\intextsep}{-1.2em}

\begin{algorithm}
\small
\caption{NIZK.Verify}
\label{ZK2}
\SetKwInOut{Input}{Input}
\SetKwInOut{Output}{Output}

 \Input{The proof $\pi=(A,B,d_1, d_2)$ and the common input $\rho=(Pk_{Sm},Pk'_{Sm}, m)$}
 \Output{bit $0/1$}
 \begin{algorithmic}[1]
     \State Compute the challenge $c=\mathcal{H}(A,B,Pk_{Sm}, Pk'_{Sm},m)$ 
     \State Check $b:g^{d_1} \overset{?}{=} A \times (\frac{Pk'_{Sm}}{Pk_{Sm}})^{c}$ and $b':g^{d_2}\overset{?}{=} B\times Pk^{c}_{Sm}$\\
     \Return $b\land b'$    
 \end{algorithmic}    
\end{algorithm}

\vspace{10pt}
\section{Application of $SPGS$ in smart environments}
\label{Application} 
Providing spontaneous join and temporary access to guests while ensuring their privacy is an overlooked problem in smart environments. In existing systems that support guests, the infrastructure management (IM) manages the join and access of the guests by dedicated guest accounts that are created ahead of time for all of them (e.g. NIST SP-1800-36 \cite{NIST24}, Kerberos). These accounts are protected by guest username and passwords that are either same for all guests which allows the guests to be anonymous but their access to the resources would be limited, or they have to contact IM to get customized guest accounts which will reveal their identity to the system. These approaches are not suitable for smart environments because of their rigidity and lack of security and privacy. We look into the problem of providing privacy preserving temporary guest accesses (spontaneously) in a smart building, by enabling a host, who is a registered entity and knows the guest, as a mediator to enroll the guest into the system. 
This scenario models many real-life settings including visitors to smart condos, work environments (such as universities and corporations), hospitals, factories, etc.

We design an {\em anonymous guest access token}, called AGAT, using  a SPGS signature scheme that allows a host to act as a sponsor and enroll a guest spontaneously to the smart building. AGAT allows the guest to show its association to a host while the identity of the guest and host (and hence their relation) remains anonymous. We propose  $k$-$AGAT$ that allows host linkability and guest linkability; the host linkability is needed to bound the number of guest accesses per host to at most $k$ (enforcing rate limit policy), and guest linkability is needed to trace the guest accesses to protect the system from the misbehaving guests who are second level members enrolled without contacting IM. Additionally,  $k$-$AGAT$ allows IM to open the tokens of misbehaving guests and identify the host who has enrolled the guest into the system to make the host accountable for their actions. Host can further help to identify the misbehaving guest.  Furthermore, a $k$-$AGAT$ token allows the host and environment both enforce policies on the guests, that is, the token issued by host also embeds some information about the guest access, such as duration of access or expiry time of token, that will be checked and enforced by the verifier (i.e., the environment).  

We note that alternative naive approaches such as providing guest certificates or simply signed guest tokens by host do not comply with our design goals. Certificates require registration and contacting certificate authorities which can become a complex and prolonged process. Also, the simple signed tokens by hosts do not allow host privacy, accountability, and bounded guest access.

In the following, we give a system and threat model for guest accesses in a smart building and define our security requirements and design goals. We then propose $k$-$AGAT$ that is used as the main building block to enable the guests join the smart building and get temporary access to the building resources.

\subsection{System and threat model}  
We consider a smart building consisting of hosts equipped with  smart devices such as mobile phones and computers, and a building with shared resources such as CCTV cameras, lightings, buzzers, printers, and so on. 
We consider a centralized architecture in which the smart building has a infrastructure management system $IM$ consisting of: 
(i)  Registration authority $RA$ that registers the hosts and performs the authentication tasks. $RA$ also acts as a token authority for $k$-$AGAT$. 
(ii) Access management authority ($AA$) which implements and enforces the access control mechanism; it checks and controls the hosts and guests accesses. $AA$ also acts as the verifier for $k$-$AGAT$. 
We assume that each resource is either capable of performing cryptographic operations, for example for verifying the tokens, or there exists an edge device that manages the resource accesses. Our goal is to enable the guests (visitors to buildings) to get spontaneous privacy-preserving temporary access to the building shared resources through hosts.

 \textbf{Threat model.} We assume the guest and host mutually trust each other in the sense that the host takes the responsibility of guest enrollment and enabling their resource accesses.  However, with respect to the IM, the guest and host may deviate from the protocol arbitrarily and collude to disrupt with the normal operation of building. 
We assume $RA$ is trusted, but $AA$ is semi-honest and wants to learn the about the host and guest identities, their accesses, an their relations.

\textbf{Security requirements and design goals.} We consider the following properties and design goals:

\begin{itemize}
    \item \textbf{Security:} Only eligible users (hosts and guests) can join and get access to the available resources temporary. 
    \item \textbf{Privacy:}  $AA$ do not learn the identities and relation of honest guests and hosts, and therefore it cannot link the access profiles to a specific host or guest identity. 
    \item \textbf{Accountability:}  The identity of the host who sponsors a misbehaving guest can be revealed.
    \item \textbf{Zero-touch enrollment:} The guests can join and get access to building shared resources spontaneously, without any required set-up or registration with $IM$.
    \item \textbf{Collaborative policy enforcement:} Both the $IM$ and host can enforce their policy preferences on the guest accesses including bounded guest accesses.
    \item \textbf{Efficiency:} The performance including the run time and storage size of the hosts, and guests are reasonable for real-world scenarios. 
\end{itemize}

\subsection{$k$-times Anonymous Guest Access Token ($k$-AGAT)}
\label{AGAT}
$k$-$AGAT$ allows a  host $H$
to enroll with a Token Authority $TA$, and  
generate a token for a guest $G$ (i.e. an unregistered entity in the system), 
such that a verifier $V$ who knows the system parameters 
can verify the guest token
without learning the actual identity of the host, the guest, and their relation. In $k$-$AGAT$ hosts can dictate their policy preferences as part of the token content $m$ including the token issuance and expiry time. Guests present the tokens to the verifiers such that tokens show the guest association with a valid host, and 
tokens can be linked together an 
\textit{counted} by  the verifier $V$. This counting allows the verifier bound the number og guest accesses per host.
In the case of any misbehavior of the host or guest,  
the token can be traced back to the host, by the $TA$, who can further help to identify the guest.

\textbf{Entities.} $k$-$AGAT$ consists of the following entities:
\\
-- \textit{Token authority (TA)} is a trusted entity in the system who runs the set up and enrolls the hosts. It can also open the tokens and identify the host issuer in case of misbehavior.\\
-- \textit{Host (H)} is a registered entity in the system who can issue tokens to guests while remaining anonymous. We assume, upon  enrollment, each host receives a private (and public) key.\\ 
-- \textit{Guest (G)} is an unregistered entity in the system which is only known to the host and it wants to get temporary access to the system and use its services. Guest and host mutually trust each other. \\
-- \textit{Verifier (V)} is a registered entity in the system who is known by the $TA$ and it receives the system public parameters,
and checks the validity of the tokens. In our setting, only the verifier who is chosen by the host can verify the validity of a guest token, but in general the tokens can be publicly verifiable.  

We note that the host and guest are not trusted with respect to the system, they may deviate arbitrarily and collude with each other. Additionally, we assume the verifier is semi-honest and wants to learn the host and guests identities, and their relations.

In the following, we give the security requirements and the $k$-$AGAT$ construction. The formal definitions and descriptions are given in Appendix \ref{AGATSec}.

 \textbf{Security requirements}
 A $k$-$AGAT$ scheme has the following properties: correctness, unforgeability, anonymity, and traceability.

\begin{itemize}
    \item \textbf{Correctness} ensures that AGAT token generation is correct, if a token that is generated by an honest host (within the $k$ bound) and  presented by the honest guest to the verifier will pass the verification.
        \item \textbf{Unforgeability} ensures that no one can issue a token on behalf of an honest host without knowing their private keys.
        Additionally, it captures the fact that the adversary cannot generate more than $k$ valid tokens that are accepted by the verifier.  
    \item \textbf{Anonymity} ensures that the identity of the host, guest, and their relation cannot be learned by the verifier.  
    We define this property as below: 
    given two guest identities $G_0$ and $G_1$ and two host identities $H_0$ and $H_1$,
    the probability of linking the guests to hosts is negligible.
Note that this definition is strong and implies both the guest anonymity and host anonymity since an adversary who can distinguish either $G_0$ from $G_1$, or $H_0$ from $H_1$ by seeing the issued and presented tokens, can also distinguish the identity of the host and guest from the challenged token, and find the the guest-host relation.
The anonymity definition does not imply host unlinkability and guest unlinkability. This is intentional since (i) the tokens issued by the same host should be linked together in order to bound the  tokens originated from the same host, and (ii) the guest is not trusted and we want to allow the verifier to trace its access behavior without knowing their identities and their relation with any given host.  
  \item \textbf{Traceability} protects the system by allowing  
  the token authority $TA$ to reveal the actual identity of the host $H$ who has issued the guest token. 
\end{itemize}

\textbf{Our construction} is generic and
 uses the generic $SPGS$ construction proposed in Section \ref{SPGSconstruction}, 
and a IND-CPA secure public key encryption scheme $E$ as its building blocks (please see Section \ref{Prelim} and \ref{SPGS} for their algorithms and security properties). 
 
In the nutshell, in our construction of $k$-$AGAT$, host act as the sponsor and guest acts as the sponsored-member of a SPGS scheme. Host uses a sponsor signature to sign the public key of the guest (which is a commitment to a random value that is used as a private key by guest) and the token contents, and issues a guest token. The guest uses a sponsored-member signature and transforms the token to a new one by proving that it knows the private key. 

The important part in our scheme is that, 
 SPGS provides complete anonymity that captures unlinkability for sponsors, but $k$-$AGAT$ does not ensure complete anonymity, it allows host linkability. For this, we follow the approach of \cite{belenkiy2009randomizable} and use pseudonyms together with our sponsored group signature to relax the anonymity of SPGS. Pseudonyms are generated by $TA$ and shared with hosts during their enrollment.
 Host will encrypt their pseudonym when it issues a token. The verifier can decrypt the pseudonym ciphertext 
 and link the tokens originated from the same host (without learning the host identity). The scheme is summarized as below (see Appendix \ref{AGATSec} for details and formal description):
 
 \begin{itemize}
     \item \textbf{Setup phase.} $TA$ maintains a pseudonym list $L_w$ that initially consists of dummy values. It also generates the SPGS private and public parameters.
     \item \textbf{Host enrollment} When a host joins, $TA$ and host run the SPGS join algorithm for a sponsor, at the end of which host receives valid private and public keys. $TA$ chooses a pseudonym $w_h$ for the host and adds it to the list $L_w$. After a batch of pseudonyms are added, it  shares the list with the the verifier. Verifier do not see the link between pseudonyms and identity of the hosts and only use the list $L_w$ to check the validity of the pseudonyms.
     \item \textbf{Guest token issuance.} To issue a token, the host and guest interact according to the SPGS join algorithm for a sponsored-member which outputs a private and public key to the guest. Additionally, host decides on the token contents $m$ such as expiry time, and  signs $m$ (that includes the public key of the guest) using the sponsor signature algorithm of SPGS. It also encrypts their pseudonym $w_h$ using the public key of the verifier. Note that to bind the SPGS signature and the pseudonym ciphertext and prevent mix and match attacks, the host will  choose a random value $r$ and  use it twice: (i) it concatenates $r$ with $m$ and generate the sponsor signature on $m||r$, and (ii) encrypts $w_h||r$. Both the sponsor signature and the pseudonym ciphertext will be included in the guest token.
     \item \textbf{Guest token presentation.} The guest will receive the guest token and transform the token to a new one by signing the token using the sponsored-member signature, and present the transformed token to the verifier. 
     \item \textbf{Guest token verification.} The verifier will verify the SPGS signatures, decrypts the ciphertext and checks whether $w_h$ is among the valid pseudonyms $w_h \in L_w$ and if it is valid it keeps a counter for $w_h$ to count the number of issued tokens by the same host. If the guest tokens originated from $w_h$ exceeds the limit $k$ (or the value less than $k$ determined by the host within the token) the verifier rejects the token.
 \end{itemize}
 
 Although to relax the anonymity one can use a different approach and let the hosts commit to their pseudonyms and prove their validity to the verifier through NIZK\footnote{This approach is publicly verifiable and one can use it to relax the assumption about the verifier's trustworthiness.} without sharing the list $L_w$ with the verifier \cite{belenkiy2009randomizable}, we chose to share $L_w$ with the verifier directly for checking the validity of pseudonyms to keep our construction simple and efficient. We show that as long as the verifier is trusted this approach ensures a secure $k$-$AGAT$ construction with the tradeoff that the pseudonyms should be added in batches and this can create a delay between when the host enrollment and when it can issue valid tokens. Since we do not have any restriction on hosts' enrollment time this approach is reasonable.

\textbf{Concrete construction.}
To instantiate our construction, we use the SPGS  concrete construction of section \ref{SPGSInit}. 
For the encryption scheme $E$, we  use the EC-based Elgamal encryption scheme.

\textbf{Security analysis.}
Below we give a theorem and a proof sketch for our construction of $k$-$AGAT$.
\begin{theorem}
    Our $k$-$AGAT$ construction ensures correctness, unforgeability, $\frac{1}{b_n}$-anonymity, and traceability assuming $SPGS$ satisfies correctness, non-frameability, sponsor anonymity, sponsored-member privacy, sponsor traceability, and $E$ is an IND-CPA secure encryption scheme.
\end{theorem} 

\textit{Proof sketch.} We omit the proof due to the space and just give the informal arguments. \\
\textit{Correctness} is satisfied due to the  correctness guarantee of the sponsored group signature scheme $SPGS$ and the public key encryption scheme $E$.
\textit{Unforgeability} follows from (i) the non-frameability of the sponsored group signature scheme which prevents the adversary from generating a valid sponsor signature $\sigma$ for a guest token  without knowing the private key of the host, (ii) the randomness $r$ used in both signature $\sigma$ and ciphertext $C$ prevents the adversary from mixing and matching 
different tokens to form a new token. 
\textit{Anonymity} follows from (i) anonymity of the sponsored group signature scheme which does not let the verifier learn whether the sponsor signature $\sigma$ in the token has been generated by host $H_0$ or $H_1$, (ii) sponsored-member privacy of SPGS which prevents the verifier to learn the identity of the guest from the sponsored-member signature used in the token, (iii) the IND-CPA security  of $E$ which does let the verifier to learn the pseudonym of the host, 
(iv) the fact that the  registered pseudonyms are independent of the real identity of the hosts, (v) pseudonyms are added to $L_w$ in batches of size $b_n$ and the verifier who controls the join of hosts and sees the latest list $L_w$ cannot guess the pseudonym of an honest host with probability greater than $\frac{1}{b_n}$. 
\textit{Issuer traceability} is ensured due to the traceability of SPGS.

\subsection{Secure guest access protocol} 

Our protocol uses the $k$-$AGAT$ construction presented in previous section directly which will be integrated to ACE-OAuth \cite{ACE-OAuth} as below. 

\begin{itemize}
    \item The host enrolls with $IM$ (more precisely $RA$) and receives the private and public parameters of the system and installs the required software and configurations by authenticating itself using the sponsor signature.
    \item When a guest visits the building, the host provides the url to install the required software and guest configuration which are stored in a public repository. The host then generates a $k$-$AGAT$ guest token by interacting with the guest.
    \item The guest sends its access request to $AA$ (which is the authorization server in ACE-OAuth) together with the transformed guest token it has received from the host. 
    \item $AA$ checks the access rights of the guest and also the access policy preferences dictated in the token by the host, such as the expiry time and the access limit. If the guest token is verified, the guest receives an access token (a signed message by the service provider, also called a proof of possession (PoP)), that it can send to the resource or edge device to get access to the resource. 
\end{itemize}

\noindent Note that until the guest token is valid the guest can send requests to $AA$. Depending on the access requests and the resources availability, the PoP can have different and possibly shorter expiry time.

\textbf{Security and privacy analysis.}
Below, we give an informal analysis of the security and privacy requirements of the scheme.

\begin{itemize}
   \item \textbf{Security:} Only eligible hosts who have registered with the RA and guest who are associated with a registered host can access the resources. This is because of non-frameability of sponsor signatures in SPGS (that is used for host authetication), and due to the unforgeability of $k$-$AGAT$, which ensures that only registered hosts who are first-level members of the building can generate valid tokens for guests (that will be accepted). 
    \item \textbf{Privacy:}  $AA$ do not learn the identities and relation of honest guests and hosts due to the anonymity of $k$-$AGAT$ and sponsor anonymity of SPGS. 
    \item \textbf{Accountability:}  The traceability of $k$-$AGAT$ ensures that the identity of the host associated with the misbehaving guest is revealed. 
    \item \textbf{Zero-touch enrollment:} The guests join through the host, and they do not need to interact with $IM$ for registration. Also, they do not require any trusted setup.
    \item \textbf{Collaborative policy enforcement:} Host can encode its preferences of access rights as part of the token content, which will be checked by $AA$. In general, $AA$ determines the guest access rights and also the policies that host can choose from; for example the validity duration of the token, where $AA$ can determine the maximum duration, and host can choose the exact duration. Additionally, both host and $AA$ can limit the usage of the token. The guest can use the token at most $k$ times (the host can determine a value lower than $k$ when issuing the token to the guest). 
    \item \textbf{Efficiency:} See the following section for implementation and evaluation results. 
\end{itemize}


\subsection{Implementation and evaluation}
\label{Implementation}
In the following, we give the implementation details, and evaluation results. The source codes are publicly available in our Github repository (https://github.com/shwdsun/GuestOnBoarding).

\textbf{Objective.}
Our goal is to measure the performance and overhead introduced by the cryptographic primitives used in the guest access protocol  compared to a non-private baseline scheme. The baseline scheme uses a regular digital signature to generate a token for the guest instead of using the $k$-$AGAT$ token of Section \ref{AGAT}. Specifically, we evaluate the runtime and storage for host, guest, and $IM$. 

\textbf{Implementation Setup.}
Our experiments were conducted on Ubuntu Linux 6.11.0-24-generic with Python 3.10.17 (conda-forge). The test machine features a 13th-generation Intel Core i5-13600K processor (14 cores, 20 threads, 3.0 GHz base frequency, 24 MB L3 cache) with 32 GB RAM. All implementations used Python 3.10.17 with the following cryptographic libraries: Charm-crypto 0.50\cite{charm-crypto}, PBC 0.5.14, GMP 6.3.0, and OpenSSL 3.3.2, managed in a dedicated Conda environment

\textbf{Implementation Details.} 
our cryptographic components integrates schemes from Charm-Crypto\cite{charm-crypto}( e.g., BBS, accessed through wrappers) with our own standalone implementations (i.e. ElGamal and NIZK).
For group signatures, we implemented the BBS scheme \cite{BBS} with MNT224 curve which requires the number of hosts to be predetermined - we set this to 100 for our implementation. The EC-ElGamal encryption 
implements 28-byte message padding for secure encoding. We use
ECDSA for a regular signature scheme. 
which implements a standard signature scheme with state management. For $NIZK$, we implemented Schnorr NIZK scheme for discrete log 
with Fiat-Shamir transform (see Section \ref{SPGSInit}.
All primitives (except BBS) are implemented over secp256k1 curve.\\

\begin{table}
\centering
\scriptsize
\begin{threeparttable}
\caption{Computation and Storage Cost in Setup Phase}
\label{tab:cost_comparison}
\begin{tabular}{lccc}
\hline
\textbf{Scheme} & \textbf{Role} & \textbf{Computation (ms)} & \textbf{Storage (bytes)} \\
\hline
\multirow{3}{*}{Scheme w/o Privacy} 
& Guest & - & - \\
& Host\tnote{a} & 0.359 & 138 \\
& IM\tnote{b} & 1.052 & 368 \\
\hline
\multirow{3}{*}{Our Scheme} 
& Guest & - & - \\
& Host\tnote{c} & 0.389 & 672 \\
& IM\tnote{d} & 1.213 & 14438 \\
\hline
\end{tabular}
\begin{tablenotes}
\item[a] E.KeyGen
\item[b] DS.KeyGen (×2)
\item[c] E.KeyGen, E.Dec
\item[d] E.KeyGen, DS.KeyGen (×2), GS.KeyGen, GS.Join, E.Enc, BN.Setup
\end{tablenotes}
\end{threeparttable}
\vspace{-2em}
\end{table}

\setlength{\intextsep}{-10pt}

\begin{table}
\scriptsize
\centering
\begin{threeparttable}
\caption{Computation and Storage Cost in guest access phase}
\label{tab:onboarding_comparison}
\begin{tabular}{lccc}
\hline
\textbf{Scheme} & \textbf{Role} & \textbf{Computation (ms)} & \textbf{Storage (bytes)} \\
\hline
\multirow{3}{*}{Scheme w/o Privacy} 
& Guest\tnote{a} & 0.730 & 230 \\
& Host\tnote{b} & 1.410 & 1049 \\
& IM\tnote{c} & 2.802 & 362 \\
\hline
\multirow{3}{*}{Our Scheme} 
& Guest\tnote{d} & 0.733 & 230 \\
& Host\tnote{e} & 12.347 & 1605 \\
& IM\tnote{f} & 17.378 & 412 \\
\hline
\end{tabular}
\begin{tablenotes}
\item[a] Randomize key, NIZK.Prove
\item[b] 
DS.KeyGen, DS.Sign
\item[c] 
DS.Verify, DS.Sign, NIZK.Verify
\item[d] Randomize key, NIZK.Prove
\item[e] 
DS.KeyGen, GS.Sign, E.Enc
\item[f] 
GS.Verify, E.Dec, DS.Sign, NIZK.Verify
\end{tablenotes}
\end{threeparttable}
\vspace{-2em}
\end{table}

\vspace{-5pt}

\textbf{Protocol Performance Evaluation.}
We evaluate the computational and storage cost. 
Note that all byte measurements use library serialization methods, resulting in larger sizes than theoretical values due to encoding overhead. Our evaluation reveals the performance trade-offs between privacy and efficiency across both setup and on-boarding phases.

\textit{Setup Phase:} As shown in Table~\ref{tab:cost_comparison}, the setup phase introduces minimal computational overhead but significant storage requirements for IM. The host's overhead 
is related to BBS group keys and pseudonyms. The IM's storage overhead is particularly pronounced due to maintaining member keys for all hosts in the system. Notably, both the setup time and storage for the BBS scheme exhibit approximately linear growth with the number of hosts.

\textit{Guest access phase:} Table~\ref{tab:onboarding_comparison} demonstrates the performance impact of our proposed scheme. While guest devices experience minimal overhead, both host and IM face significant computational costs due to BBS group signature operations and ElGamal encryption/decryption. In terms of storage, the host requires additional space for the larger BBS signature and group public key, while IM and the guest storage requirements remain relatively modest.

Despite these overheads, the costs remain reasonable
for practical
deployments. The millisecond-scale computation times are negligible compared to typical network latencies, and the kilobyte-scale storage are well within modern smart device capabilities.

\section{Conclusion}
\label{Conclusion}
We proposed sponsored group signature (SPGS), a variation of a group signature that supports two levels of signers. The level 1 signers (sponsors) can enroll level 2  signers (sponsored-members) while allowing both members sign anonymously and ensure their signatures remain unforgeable. In SPGS the sponsored-signature can be linked together and in case of misbehavior the identity of the sponsor can be revealed by the group manager. We formalized SPGS using the game-based model and proposed a generic construction using a commitment, a partially dynamic group signature, and a knowledge sound NIZK scheme. 
We also showed the application of SPGS in constructing a guest access token, $k$-$AGAT$, for a smart building scenario in which a host provides spontaneous privacy preserving temporary access to its guests to access the building shared resources.  The $k$-$AGAT$ construction 
$k$-$AGAT$ allows the verifier to bound the number of guest tokens that can be presented by the guest of the same host. We proposed a secure guest access protocol using  $k$-$AGAT$ and
showed its efficiency by providing a proof-of-concept implementation of its cryptographic components. 
Extending our work to a multi-building scenario and different levels of guests is an interesting future work.

\bibliographystyle{IEEEtran}
\bibliography{PoPETS-Nov30/sample-base.bib}

\appendix

\section{Group signature: security model and definitions}
\label{GSModel}
Bootle et al. \cite{bootle2020foundations} proposed a framework to model and define the security requirements of a dynamic group signature scheme. We adapt their model to our setting and define the model and security properties that are required for a group signature scheme that is used in our SPGS generic construction. Also, these definitions are used as a basis to define the SPGS security model.

The partially dynamic group signature $GS$ composed of the following algorithms: 

\begin{itemize}
    \item $GS.Setup(1^\lambda, setpp)=(pp, msk)$, is run by $\mathcal{M}$ and it takes the security parameter $\lambda$ and set up parameters $SP$, and outputs the public parameters $pp$ and the manager's secret key $msk$. 
    \item $GS.Join(id_u,param_\mathcal{M})$ is an interactive algorithm that is run between a $\mathcal{M}$ and the user $id_u$ who desire to join the group (we assume all interactions takes place over a secure channel similar to \cite{buser2019dgm}).
    \item $GS.Sign(m, param_{u})= \sigma$, this algorithm is run by the member $id_u$ with its private parameters $param_{u}$, and outputs a valid anonymous signature $\sigma$.
    \item $GS.Verify(m,\sigma,pp)=0/1$, is run by the verifier and takes as input the message $m$, the signature $\sigma$, and the public parameters of the system $pp$, and outputs 1 if the signature is valid, or 0 otherwise. 
    \item $GS.Open(msk,\sigma)=id_u$, is run by $\mathcal{M}$ to reveal the identity of the signer of $\sigma$. 
\end{itemize}

We consider that the group manager $\mathcal{M}$ act as both issuing authority and opening authority, and is assumed to be partially corrupted (i.e., its private state can be leaked).  Also, we do not consider the $Judge$ algorithm 
since $\mathcal{M}$ follows the protocol correctly and traces the signatures if needed. Additionally, we consider that the group is partially dynamic (the entities can join at any time). 

\textbf{Security requirements.}
We consider that $GS$ ensures the following security requirements:

\begin{itemize}
  \item \textbf{Correctness} ensures that if an honest member joins and generates a signature, their signature will be verified.
  In the partially dynamic setting where the members can join at any time, these two conditions should hold for all honest members under any schedule under which the members join the group. Therefore, to formalize correctness we consider an adversary who can control users' joining  process and chooses the messages and the identity of the signers. We require that for any messages and identities that the adversary chooses, the signatures generated can be verified correctly.
   \item \textbf{Non-frameability} ensures that no one can frame an honest member; that is, even if group members collude (all except at least one member) and even the group manager's private state is leaked, they cannot generate a valid signature that is opened to the honest member who has not generated the signature. \\
  This notion of non-frameability implies unforgeability of the group signature, since the adversary can corrupt and learn all keys except for the honest member that it generates the signature for.  
  \item \textbf{Anonymity} states that no one can distinguish the member who has generated a given group signature from $id_{u_0}$ and $id_{u_1}$ better than a random guess even if they have seen the output of the $Open$ algorithm for some signatures. $\mathcal{M}$ is considered trusted for anonymity as $\mathcal{M}$ can always open the signatures and learn the identity of the sponsors. This definition
  captures unlinkability as well.
  \item \textbf{Traceability} protects the group manager by ensuring that any valid signature will be opened to the identity of a group member. This is achieved even if the group manager's private state is leaked (i.e., $\mathcal{M}$ is considered partially corrupted). 
  \end{itemize}

\textbf{Oracles.} To formally define the security requirements, we use the following oracles (given in Figure \ref{GSoracles}).

\begin{itemize}
    \item $\mathbf{AddHU(id_{u})}$: This oracle allows the adversary to add honest users through honest execution of $GS.Join$ without learning their private parameters.
    \item $\mathbf{CorrU(id_{u})}$ allows the adversary to corrupt the users and learn  both their communication transcript when they run $GS.Join$ algorithm honestly and their private parameters (including signing keys). 
    \item $\mathbf{AddCU(id_{u})}$: This oracle allows the adversary to add corrupted users to the group. The adversary can deviate from the $GS.Join$ protocol and send arbitrary messages to $\mathcal{M}$ see its output $out_\mathcal{M}$. 
    \item $\mathbf{AChal_b(pp, m, id_{u_0}, id_{u_1})}$: This is a left-right oracle for defining anonymity. It takes as input the group public parameters $pp$, a message $m$, and two honest users identities $id_{u_0}$ and $id_{u_1}$, and returns a group signature
    on the message $m$ using the private parameters of $param_{u_b}$ for $b \leftarrow \{0,1\}$. The adversary can call this oracle once.
    \item $\mathbf{SignHU(m,param_{u})}$ This oracle is used by the adversary to obtain signatures for an honest user whose private parameters are not known by the adversary. It returns a signature $\sigma$ on the message $m$ using the $GS.Sign$ algorithm with the private parameters of $id_{u}$.
    \item $\mathbf{Open(msk, m, \sigma, pp)}$ returns the identity of the sponsor $id_{u}$ who has generated the group signature $\sigma$. The oracle cannot be called on a signature obtained from the $AChal_b$.
    \item $\mathbf{CorrM()}$ returns the private parameters of the manager $\mathcal{M}$.
\end{itemize}

\vspace{5pt}
\begin{figure}[!bht]

\fbox{
    \begin{minipage}{\columnwidth}
    \begin{small}
    $\mathbf{AddHU(id_{u})}$
    \hrule
    \begin{itemize}
        \item If $N>N_{max}$ Return $\perp$
        \item Update $\mathcal{H} \cup \{id_{u}\}$ and $N=N+1$
        \item $param_{u} \leftarrow GS.Join(id_{u}, param_\mathcal{M})$
        \item Let the Join transcript is stored in $trscJoin_{u}$
        \item Store $Q_{HU} =Q_{HU} \cup (id_{u},trscJoin_{u}, param_{u}, out_\mathcal{M})$
        \item Return $pp$
    \end{itemize}

   $\mathbf{CorrU(id_{u})}$
    \hrule
    \begin{itemize}
        \item If $id_{u} \notin \mathcal{H}$ Return $\perp$
        \item Update $\mathcal{C} \cup \{id_{u}\}$ and $\mathcal{H}=\mathcal{H} \backslash \{id_{u}\}$
        \item Retrieve $(id_{u},trscJoin_{u}, param_{u}, out_\mathcal{M})$ from $Q_{HU}$
        \item Return $trscJoin_{id_{u}}$, $param_{u}$, and $out_\mathcal{M}$
   \end{itemize}

    $\mathbf{AddCU(id_{u})}$
    \hrule
    \begin{itemize}
        \item If $id_{u} \in \mathcal{H}$ Return $\perp$
        \item Update $\mathcal{C} \cup id_{u}$
        \item $param_{u} \leftarrow GS.Join(id_{u}, param_\mathcal{M})$
        \item Return $param_{u}$ and $out_\mathcal{M}$
    \end{itemize}

    $\mathbf{AChal_b(pp, m, id_{u_0}, id_{u_1})}$
    \hrule
    \begin{itemize}
        \item If $id_{u_0}, id_{u_1}$ $\notin \mathcal{H}$ Return $\perp$
        \item $ \sigma_b \leftarrow GS.Sign(m, param_{u_b})$
        \item If $GS.Verify(m, \sigma_b, pp)=0$ Return $\perp$
        \item Update $Q_{AChal} \cup (m, \sigma_b)$
        \item Return $\sigma_b$
    \end{itemize}
%
     $\mathbf{SignHU(m,param_{u})}$
     \hrule
     \begin{itemize}
        \item If $param_{u}= \perp$ Return $\perp$
        \item $ \sigma \leftarrow GS.Sign(m, param_{u})$ 
        \item Update $Q_{Sign} \cup (m,\sigma)$
        \item Return $\sigma$
    \end{itemize}

    $\mathbf{Open(msk, m, \sigma, pp)}$
     \hrule 
     \begin{itemize}
        \item If $(m, \sigma) \in Q_{Sign} \cup Q_{AChal}$ Return $\perp$
        \item If $GS.Verify(m, \sigma, pp)=0$ Return $\perp$
        \item Return $GS.Open(msk, \sigma)$
    \end{itemize}

     $\mathbf{CorrM()}$
     \hrule 
     \begin{itemize}
        \item Return $param_{\mathcal{M}}$
    \end{itemize}

   \end{small} 
     \end{minipage}
}
     \caption{Oracles used in experiments}
     \label{GSoracles}
\end{figure}

\vspace{1em}

\begin{definition}
\label{GSDefinition}
    For any security parameter $\lambda \in \mathrm{N}$ and for any PPT adversary $\mathcal{A}$, we say that $SPGS$ provides:
    \begin{enumerate}
        \item \textbf{Correctness} if there exists a negligible function $\nu_1$ such that 
        $Adv^{Corr}_{GS,\mathcal{A}} (\lambda)=Pr[Exp^{Corr}_{GS, \mathcal{A}}(\lambda)=1] \geq 1-\nu_1(\lambda)$\\

        \item \textbf{Non-frameability} if there exists a negligible function $\nu_2$ such that
        $Adv^{Non-Frame}_{GS,\mathcal{A}} (\lambda)=Pr[Exp^{Non-Frame}_{GS, \mathcal{A}}(\lambda)=1] \leq \nu_2(\lambda)$\\

        \item \textbf{Anonymity} if there exists a negligible function $\nu_3$ such that
        $Adv^{Anon}_{GS,\mathcal{A}} (\lambda)=Pr[Exp^{Anon-b}_{GS, \mathcal{A}}(\lambda)=1] \leq \frac{1}{2}+\nu_3(\lambda)$\\

        \item \textbf{ Traceability} if there exists a negligible function $\nu_4$ such that
        $Adv^{Trace}_{GS,\mathcal{A}} (\lambda)=Pr[Exp^{Trace}_{GS, \mathcal{A}}(\lambda)=1] \leq \nu_4(\lambda)$\\
    \end{enumerate}

    where $Exp^{Corr}_{GS, \mathcal{A}}$, $Exp^{Non-Frame}_{GS, \mathcal{A}}$, $Exp^{Anon}_{GS, \mathcal{A}}$, $Exp^{Trace}_{GS, \mathcal{A}}$, are defined in Figure \ref{GSexperiments}.\\
\end{definition}

\begin{figure}[!bht]
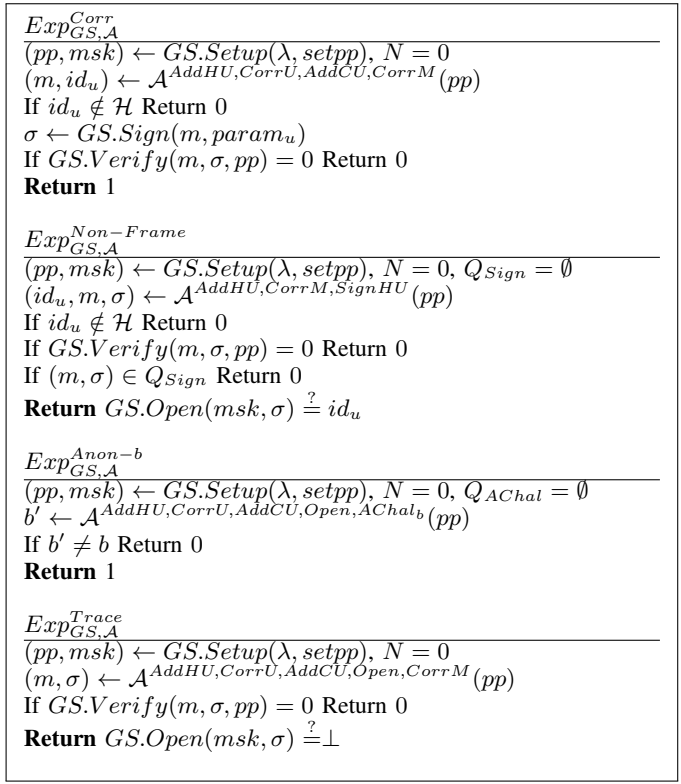


\fbox{
    \begin{minipage}{0.95\columnwidth}
    \begin{small}
    $Exp^{Corr}_{GS, \mathcal{A}}$
    \hrule 
    $(pp,msk) \leftarrow GS.Setup(\lambda,setpp)$, $N=0$\\
    $(m, id_{u}) \leftarrow \mathcal{A}^{AddHU,CorrU,AddCU,CorrM}(pp)$\\
    If $id_{u} \notin \mathcal{H}$ Return $0$\\
    $\sigma \leftarrow GS.Sign(m,param_{u})$\\
    If $ GS.Verify(m,\sigma,pp)=0$ Return $0$\\
    \textbf{Return} 1\\ 

    $Exp^{Non-Frame}_{GS, \mathcal{A}}$
    \hrule 
    $(pp,msk) \leftarrow GS.Setup(\lambda,setpp)$, $N=0$, $Q_{Sign}=\emptyset$\\
    $(id_{u},m,\sigma) \leftarrow \mathcal{A}^{AddHU,CorrM,SignHU}(pp)$\\
     If $id_{u} \notin \mathcal{H}$  Return $0$\\
    If $GS.Verify(m,\sigma,pp)=0$ Return $0$\\
    If $(m, \sigma) \in Q_{Sign}$ Return $0$\\
    \textbf{Return} $GS.Open(msk, \sigma) \overset{?}= id_{u}$\\

    $Exp^{Anon-b}_{GS, \mathcal{A}}$
    \hrule 
    $(pp,msk) \leftarrow GS.Setup(\lambda,setpp)$, $N=0$, $Q_{AChal}=\emptyset$\\
     $b' \leftarrow \mathcal{A}^{AddHU,CorrU,AddCU,Open,AChal_b}(pp)$\\
     If $b' \neq b$ Return $0$\\
    \textbf{Return} 1\\ 

     $Exp^{Trace}_{GS, \mathcal{A}}$
    \hrule 
    $(pp,msk) \leftarrow GS.Setup(\lambda,setpp)$, $N=0$\\
    $(m,\sigma)  \leftarrow \mathcal{A}^{AddHU,CorrU,AddCU,Open,CorrM}(pp)$\\
     If $GS.Verify(m,\sigma,pp)=0$ Return $0$\\
    \textbf{Return} $GS.Open(msk, \sigma) \overset{?}= \perp$\\
    
    \end{small}
    \end{minipage}
    }
     \caption{Security games of group signature $GS$}
     \label{GSexperiments}
\end{figure}

\vspace{1em}

\section{Non-interactive zero knowledge arguments of knowledge (NIZK): security model and definitions}
\label{NIZKModel}
 NIZK allows a prover who knows a witness $w$ for a statement $x$ can convince a verifier that $(w,\;x) \in R$, without revealing the witness $w$, where $R$ is a relation ($R\subset \mathcal{L})$ defined in language $\mathcal{L}$.
NIZK consists of the following algorithms:
\begin{itemize}
    \item $NIZK.Setup(R)$ which takes the relation $R$ (which implicitly defines the security parameter $\lambda$) and outputs the common reference string $crs$. 
    \item $NIZK.Prove(crs, \rho, x, w)$ which takes the common reference string $crs$, the common input $\rho$, the statement $x$ that is being proven about $\rho$, and the witness $w$ as input and outputs the proof $\pi$. 
    \item $NIZK.Verify(crs, \pi, \rho, x)$ receives the common reference string $crs$, the proof $\pi$, the common input $\rho$, and the statement being proven as input and outputs $1$ if the proof is verified and $0$ otherwise.
\end{itemize}

\textbf{Security requirements.} We consider a NIZK that ensures perfect completeness, computational zero-knowledge, and computational knowledge-soundness properties\cite{benhamouda2023anonymous}. 

\textbf{Perfect completeness} ensures that an honestly generated proof $\pi$ for a true statement $(w,\;x) \in R$ will pass the verification. That is, for all relations $R$ and $(w,\;x) \in R$, for all $ crs \leftarrow NIZK.Setup(R)$, and honestly generated proofs $\pi \leftarrow NIZK.Prove(crs, \rho, x, w)$, the verification passes:
 $1\leftarrow NIZK.Verify(crs, \pi, \rho, x)$

 \textbf{Computational zero-knowledge} ensures that the proof $\pi$ does not reveal any additional
 information about the witness $w$ other than the fact that the statement is correct, i.e., $(w,\;x) \in R$. Formally, this means that for all relations $R$ and all PPT
 adversaries $\mathcal{A}$, there exists a simulator $Sim$ that outputs the public parameters and a trapdoor $td$ which allows to verify the proofs without knowing the witness.
 
\noindent
\begin{footnotesize}
\begin{flushleft}
  $\Bigg| Pr\left[ \begin{array}{l}
  crs \leftarrow NIZK.Setup(R);\\
   1\leftarrow \mathcal{A}^{NIZK.Prove(crs,\rho,x, w)}(crs) \end{array}\right] - $ \\ $Pr\left[ \begin{array}{l}
  (crs,td) \leftarrow Sim(R);\\
   1\leftarrow \mathcal{A}^{Sim(td,\rho,x, \cdot)}(crs) \end{array}\right]\Bigg| \leq \nu(\lambda) $
   \end{flushleft}
\end{footnotesize}

\textbf{Computational Knowledge Soundness} ensures that (i) for all relations $R$ and for every PPT adversary $\mathcal{A}$ there exists an extractor $Extract$ that can extract a valid witness for every
 valid proof $\pi$ that the adversary generates 
 and (ii) the extractor-generated $crs$ is indistinguishable from an honestly generated
$crs$ (this is called set-up indistinguishability). Our definition follows from \cite{benhamouda2023anonymous} which removes the access to a simulation oracle in \cite{bernhard2012not}, and simplify it by
 allowing the extractor to select the random coins of the adversary, and allowing the use of a $crs$.

{\em Extractability:}

\noindent
 \begin{footnotesize}
\begin{flushleft}
$Pr\left[ \begin{array}{l}
  (crs,td) \leftarrow Extract(R);\\
  r \leftarrow  rnd_\mathcal{A};\\
   (x, \pi)\leftarrow \mathcal{A}(crs,r)\\ w \leftarrow Extract^{\mathcal{A}}(td,x,\pi,r)\end{array} \Bigg| \begin{array}{l} 1 \leftarrow NIZK.Verify(crs, \pi, \rho, x) \land\\ (w,x) \notin R \end{array}
   \right]$ \\  $\leq \nu(\lambda)$
\end{flushleft} 
\end{footnotesize}

where $rnd_\mathcal{A}$ is the distribution of random coins of the adversary.

{\em Statistical set-up indistinguishability:} for arbitrary adversary $\mathcal{B}$ (even non-polynomial time)

\noindent
\begin{footnotesize}
  $\Big|Pr\left[ \begin{array}{l}
  crs \leftarrow NIZK.Setup(R)\big|
   1\leftarrow \mathcal{B}(crs)\end{array}\right] - \\  Pr\left[ \begin{array}{l}
  (crs,td) \leftarrow Extract(R)\big|
   1\leftarrow \mathcal{B}(crs) \end{array}\right]\Big| \leq \frac{1}{2}+\nu(\lambda) $
\end{footnotesize}

\section{Signature of knowledge (SOK)}
\label{SOK}
SOK allows a signer who knows a valid witness $w$ for a statement $x \in L$, for the $NP$ language $L$, such that $M_L(x,w)=accept$, and $M_L$ is a polynomial time Turing machine, to generate a signature on message $m$. A signature of knowledge (SOK) should ensure SimExt-security which implies correctness, simulatability and exractability. It was shown in \cite{chase2006signatures} that using a dense public key cryptosystem\footnote{A dense public key cryptosystem is a public key cryptosystem consisting of $(KeyGen, Enc, Dec)$ algorithms with two differences: (i) the public key output by the $KeyGen$ algorithms is indistinguishable from a uniform distribution, and (ii) Rather than for all public keys, for only a sufficient set of public keys the security properties hold.} \cite{de1992zero} (which is equivalent to a commitment scheme\footnote{ An extractable commitment scheme is a commitment scheme that outputs a trapdoor together with system parameters, such that an extractor can extract the committed value using the trapdoor. It ensures extractability in addition to hiding and binding properties. Extractability is defined as follows: there exists an extractor who can extract the value it has been committed knowing the public random string an the auxiliary data.}  \cite{de2000necessary}) and a simulation-sound non-interactive zero knowledge proof, one can construct a signature of knowledge for any language $L$.  The sponsored-member signature in our generic construction can be seen as a group signature and a SOK (with weaker security than \cite{chase2006signatures}). One can also use the generic SOK construction of \cite{chase2006signatures}, but we note that the SOK of \cite{chase2006signatures} has stronger security properties that is needed for SPGS. Our constructions use minimal security assumptions.

\section{Linkability definition of sponsored-member signature}
\label{Linkability}
\begin{lemma}
\label{LemmaLinkiability}
    The sponsored signature linkability of SPGS (defined by $Exp^{SmLink-b}_{SPGS,\mathcal{A}}$) is equivalent to the signer linkability notion defined for LSAG (the linkable spontaneous anonymous group signature) of \cite{liu2004linkable}.
\end{lemma}

\textit{Proof.} We argue that the two definitions are equivalent, meaning that if the Liu et al. linkability definition \cite{liu2004linkable} holds, our sponsored-member linkability definition will hold and vice versa.

Let's consider that there is an SPGS scheme which ensures linkability notion as defined by statements (i) and (ii) above. Therefore, there exists a PPT algorithm $F$ that can link the two signatures with high probability. We show that this scheme also ensures the linkability notion of Definition 1.5. To find the advantage of $\mathcal{A}$ in $Exp^{SmLink-b}_{SPGS,\mathcal{A}}$, we assume $\mathcal{A}$ uses the algorithm $\mathcal{F}$ as a subroutine, it sends the messages and signatures received by $LChal_b$ to $\mathcal{F}$ and outputs $b'=0$ if $\mathcal{F}$ outputs $1$, and outputs $b'=1$ if $F$ outputs $0$. Note that if $\mathcal{F}$ can link the two signatures that are issued by the same signer with non-negligible probability (i.e., corresponding to the probability statement i) then if challenger chooses $set_0$ then the probability that $F$ will output $1$ will be  non-negligible. Also, if  $\mathcal{F}$ always outputs $0$ for the two signatures that are issued by different signers with non-negligible probability (i.e., corresponding to the probability given in  statement ii) then if challenger chooses $set_1$ according to our definition then the probability then $F$ will output $0$ will be non-negligible. Therefore, the advantage of $\mathcal{A}$ is equal to advantage of $F$ which is non-negligible (i.e. $1-\nu(\lambda)$).

The reverse side of the statement state that if there is an algorithm $\mathcal{A}$ that can win the $Exp^{SmLink-b}_{SPGS,\mathcal{A}}$ with non-negligible probability, then there exists an algorithm $F$ that can link the two signatures with non-negligible probability. We use contradiction to prove this statement. Let's consider that the algorithm $F$ will link the two signatures with only negligible probability, then it is clear that $\mathcal{A}$ which uses $\mathcal{F}$ as a subroutine will succeed with negligible probability which contradicts our assumption.

\section{Security analysis of SPGS construction} 
\label{Proofs}
The proof of theorem consists of a set of lemmas, one for each property, that are outlined below.

\begin{lemma}
    \textbf{Correctness.} The construction of $SPGS$ given in Section \ref{SPGSconstruction} achieves correctness,  if the group signature $GS$ and commitment $C$ ensure correctness and $NIZK$ ensures completeness.
\end{lemma}
 The proof is straightforward. If $GS$ ensures correctness, then the sponsor signature $\sigma$ will be correct. Also, if the commitment scheme $C$ ensures correctness the the proof $\pi$ can be constructed correctly for the $PK_{Sm}$. Additionally, the completeness property of the $NIZK$ ensures that any proof generated by the honest parties will be verified. Therefore, the sponsored signature $\sigma'$ will be verified.

\begin{lemma}
    \textbf{Non-frameability.} The construction of $SPGS$ given in Section \ref{SPGSconstruction} achieves non-frameability,  if $NIZK$ ensures knowledge soundness, the commitment $C$ is binding, and the group signature $GS$ is non-frameable.
\end{lemma}

Depending on $attr$, We consider two cases: (i) SP-non-frameability, where $attr=Sp$, and the adversary $\mathcal{A}$ in game $Exp^{Non-Frame}_{SPGS,\mathcal{A}}$ after issuing oracle queries outputs $(Sp, id_{Sp}, m, \sigma_{Sp})$, and (ii) Sm-non-frameability, where $attr=Sm$, and the adversary $\mathcal{A}$ in game $Exp^{Non-Frame}_{SPGS,\mathcal{A}}$ after issuing oracle queries outputs $(Sm, id_{Sm}, m, \sigma_{Sm})$. \\

\noindent \textbf{(i) Sp-non-frameability.} In this game, the adversary $\mathcal{A}$ can only issue queries to $SignHU(Sp, \cdot,\cdot)$ oracle and at the end it wins if it outputs a valid sponsor signature for an honest member which can be opened to the sponsor correctly. The sponsor signature $\sigma_{Sp}$ in our scheme is generated through a group signature $GS$, and thus the $Sp$-non-frameability of $SPGS$ directly reduces to the non-frameability of $GS$. In another words, we show that if $SPGS$ does not ensure $Sp$-non-frameability then we can construct an adversary $\mathcal{A}$ who can break the non-frameability of $GS$. Let's consider that there exists an adversary $\mathcal{B}$ who can output $(m,\sigma)$ that opens to $id_{Sp}$ in $Exp^{Non-frame}_{SPGS,\mathcal{B}}$ with non-negligible probability. Then  $\mathcal{A}$  can use  $\mathcal{B}$ as a subroutine to break the  non-frameability of $GS$. $\mathcal{A}$ will answer the $AddHU(Sp,\cdot)$, $CorrM$, and $SignHU(Sp,\cdot,\cdot)$  queries of $\mathcal{B}$ in $Exp^{Non-frame}_{SPGS,\mathcal{B}}$ experiment  by using $AddHU$, $CorrM$, and $SignHU$ oracles of $GS$ (see $Exp^{Non-frame}_{GS,\mathcal{A}}$ in Appendix \ref{GSModel}). When $\mathcal{B}$ outputs $(Sp, m, \sigma)$, $\mathcal{A}$ outputs $(m, \sigma)$. The success probability of $\mathcal{A}$ is the same as the success probability of $\mathcal{B}$. This contradicts the assumption that $GS$ is non-frameable, and hence we conclude that $SPGS$ ensures $Sp$-non-frameability.\\

\noindent \textbf{(ii) Sm-non-frameability.}  In this game, the adversary $\mathcal{A}$ can issue queries to both  $SignHU(Sp, \cdot,\cdot)$ and $SignHU(Sm, \cdot,\cdot)$ oracles, and it wins if it can output a valid signature $\sigma_{Sm}=\sigma^{*}$ such that it is opened to the actual sponsor. We prove Sm-non-frameability using the hybrid game as below:

\begin{description}
    \item[$Hyb_0$.] This game corresponds to the non-frameability game of $SPGS$.
    \item[$Hyb_1$.] This game is similar to the previous one except that for any query to $SignHU$ oracle, $NIZK$ set up is done by the extractor from the knowledge soundness. Due to the set up indistinguishability from knowledge soundness this game is indistinguishable from the previous game.
    \item[$Hyb_2$.] This game is similar to the previous one except that after the adversary outputs a forgery $(m^*, \sigma^*)$ where $\sigma^*=(\sigma,Pk^*_{Sm},\pi)$, the game aborts if the extractor aborts and cannot extract a valid witness from $\pi$. This game is computationally indistinguishable from the previous game under the extractability of the $NIZK$ (the abort event happens only with negligible probability).
    \item[$Hyb_3$.] This game is similar to the previous one except that the game aborts if the extracted witness does not match the forgery $\sigma^*=(\sigma,Pk^*_{Sm},\pi)$. Let's consider that the witness is $(sk^*_{Sm}, r)$. Therefore, we should have  
    $(Pk^*_{Sm}, r)=C.Com(sk^*_{Sm})$. This game is indistinguishable from the previous game since $C$ ensures binding property, and the probability that any of these checks fails is negligible. Additionally, if any of these statements are not satisfied then the proof $\pi$ will not be verified and the extractor aborts.
     \item[$Hyb_4$.] Let's consider the forgery is $\sigma^*=(\sigma,Pk^*_{Sm},\pi)$ on message $m^*$. This game is similar to the previous one except that the game aborts if $(m^*, (\sigma, Pk^*_{Sm},\cdot))  \notin Q_{Sign}$. This game is indistinguishable from the previous game since if this is the case then $\sigma$ should have been forged, which only happens with negligible probability due to the non-frameability of $GS$. 
\end{description}

We can see that assuming the $NIZK$ ensures knowledge-soundness, and commitment $C$ is binding, we can reduce the Sm-non-frameability of the $SPGS$ to the last hybrid game. This shows that if $\mathcal{A}$ outputs a forgery $\sigma^*=(\sigma,Pk^*_{Sm},\pi)$ on $m$ then we can construct an attacker $\mathcal{B}$ which can use the $\mathcal{A}$ as subroutine to output a valid forgery for the group signature $\sigma$ on $m=Pk^*{Sm}$. This concludes the proof.

\begin{lemma}
    \textbf{Sponsor anonymity.} The construction of $SPGS$ given in Section \ref{SPGSconstruction} achieves sponsor anonymity,  if 
    $GS$ ensures anonymity for the signer. 
\end{lemma}

We give the hybrid game as below:

\begin{description}
    \item[$Hyb_0$.] This game corresponds to the sponsor anonymity game of $SPGS$. 
    \item[$Hyb_1$.] This game is similar to the previous one except that for any choice of $b \in \{0,1\}$, $\sigma'_b$ is generated by $GS$ using private parameters of a random entity $id'_{Sp}$. This is indistinguishable from the previous game since $GS$ satisfies anonymity.
\end{description}

We can see that
the last game does not have anything related to the identity of the actual sponsor $id_{Sp}$ and hence it is indistinguishable from the random except with a negligible probability.  This concludes the proof.

\begin{lemma}
    \textbf{Sponsor traceability.} The construction of $SPGS$ given in Section \ref{SPGSconstruction} achieves sponsor traceability,  if the group signature $GS$ ensures traceability.
\end{lemma}
We show that if $GS$ ensures traceability, then our construction of $SPGS$ ensures traceability. Note that $Open$ algorithm only uses the sponsor signature (computed through $GS$) to identify the sponsor identity. Let's assume that $SPGS$ does not ensure sponsor traceability and there exists an adversary  $\mathcal{B}$ that can output $(SP, m, \sigma)$ such that $SPGS.Open^{Sp}(msk,\sigma)=\perp$, or outputs $(Sm, m \sigma)$ such that $SPGS.Open^{Sm}(msk,\sigma)=\perp$ with non-negligible probability. Then we can construct the adversary $\mathcal{A}$ who can use $\mathcal{B}$ as a subroutine to break the traceability of $GS$. $\mathcal{A}$ will respond to the $AddHU(Sp, \cdot)$,$CorrU(Sp,\cdot)$, $AddCU(Sp,\cdot)$, $Open(Sp,\cdot)$ and $CorrM$ queries of $\mathcal{B}$ in $Exp^{SpTrace}_{SPGS,\mathcal{B}}$ using the $AddHU$, $CorrU$, $AddCU$, $Open$ and $CorrM$ oracles of $GS$ (see $Exp^{Trace}_{GS,\mathcal{A}}$ in  Appendix \ref{GSModel}) respectively. When $\mathcal{B}$ outputs $(Sp, m, \sigma)$, $\mathcal{A}$ outputs $(m ,\sigma)$. If $\mathcal{B}$ outputs $(Sm, m, \sigma)$, $\mathcal{A}$ parses $\sigma=(\sigma', m, m', \pi)$, and outputs $(m ,\sigma')$. The success probability of $\mathcal{A}$ in this game is the same as the success probability of $\mathcal{B}$. This contradicts the assumption that $GS$ ensures traceability, and hence we conclude that SPGS ensures sponsor traceability.\\

\begin{lemma}
    \textbf{Sponsored-member privacy.} The construction of $SPGS$ given in Section \ref{SPGSconstruction} achieves sponsored-member privacy,  if $NIZK$ is a zero knowledge argument and the commitment scheme $C$ is hiding. 
\end{lemma}

We give the hybrid game as below:

\begin{description}
    \item[$Hyb_0$.] This game corresponds to the sponsored-member privacy game of $SPGS$.
    \item[$Hyb_1$.] This game is similar to the previous one except that for any query to $SignHU$ oracle and for signature $\sigma_b=(\sigma',Pk_{Sm},\pi)$ sent to the adversary, the proof $\pi$ is generated through the $NIZK$ simulator. This game is indistinguishable from the previous one because of $NIZK$ is zero-knowledge. 
    \item[$Hyb_2$.] This game is similar to the previous one except that for any query to $SignHU$ oracle and for signature $\sigma_b=(\sigma',Pk_{Sm},\pi)$ sent to the adversary, $Pk_{Sm}$ is chosen randomly, and the proof $\pi$ is generated through the $NIZK$ simulator. This game is indistinguishable from the previous one because the commitment $C$ that is used to generate $Pk_{Sm}$ is hiding. 
\end{description}

We can see that assuming the $NIZK$ is zero-knowledge, and commitment $C$ is hiding, we can reduce the sponsored-member privacy of the $SPGS$ to the last hybrid game. The last game does not have reveal any information about the identity of the sponsored-member  and hence it is indistinguishable from the random except with a negligible probability.  This concludes the proof.

\begin{lemma}
    \textbf{Sponsored-member linkability.} The construction of $SPGS$ given in Section \ref{SPGSconstruction} achieves sponsored-member linkability,  if 
    $GS$ ensures non-frameability.
\end{lemma}

We give the hybrid game as below:

\begin{description}
    \item[$Hyb_0$.] This game corresponds to the sponsored-member linkability game of $SPGS$. 
   \item[$Hyb_1$.] This game is similar to the previous one except that when the adversary receives two signatures, i.e., $\sigma_b=(\sigma'_b,Pk^b_{Sm},\pi_b)$ for $b\in \{0,1\}$, it compares the two signatures $\sigma'_b$ and outputs 1 if $\sigma'_0 = \sigma'_1$ and outputs 0 if $\sigma'_0 \neq \sigma'_1$. 
   If  $\sigma'_0 \neq \sigma'_1$ and one (or both) $\sigma'_b$ have not been queried to $SignHU$ oracle before, meaning that $(\sigma'_b,Pk_{b,Sm},\cdot) \notin Q_{Sign}$ the adversary aborts. This game is indistinguishable from the previous one since if $\sigma'_b$  has not been queried before, then it should have been forged which happens only with negligible probability since $GS$ is non-frameable.
\end{description}

We can see that
the last hybrid game shows  the adversary can output the bit $b$ correctly with high probability. Otherwise, one can construct an adversary $\mathcal{B}$ which can output a valid forgery for the group signature $\sigma'$ and break its non-frameability. This concludes the proof. 

\section{Security definition of AGAT}
\label{AGATSec}

 We define $k$-$AGAT$ with the following algorithms:

\begin{itemize}
    \item  $AGAT.Setup(\lambda)$ is run by $TA$ which takes the security parameter $\lambda$ and outputs the system public parameters $pp$ and the private parameters $Priv_{TA}$ for the $TA$. $pp$ includes the registered verifiers' public parameters, the threshold value $k$, the show function $f()$.
\item $AGAT.Enroll(pp,\;H,\;Priv_{TA})$ is an interactive algorithm run between the $TA$ and the host $H$. It takes $pp$,  the identity of the host $H$, and the private parameters of $TA$, $Priv_{TA}$, as input, and outputs the host's private parameters  $Priv_h$ that includes a private key 
$sk_h$ 
to the  host $H$, and success $\top$ or failure $\perp$ to $TA$. 
\item $AGAT.Issue(pp,\; G,\; m,,\; Pub_V,\;Priv_h)$\footnote{This algorithm can be defined as a non-interactive algorithm as well. } 
is an interactive algorithm run by the host $H$ and guest $G$, and takes $pp$,  the identity of the guest $G$, the message $m$, the verifier's public parameters $Pub_V$, and the host's private parameters $Priv_h$, 
and outputs a  token $\mathcal{T}$ on $m$ together with private parameters $Priv_g$ to the guest, and success $\top$ or failure $\perp$ to $H$. 
\item $AGAT.Show(pp,\;\mathcal{T},\; Priv_g)$ is run by the guest $G$ and takes $pp$, the token $\mathcal{T}$, and the private parameters of the guest $Priv_g$,
and outputs a token $\mathcal{T}'$  which will be presented to $V$. $\mathcal{T'}=f(\mathcal{T})$, and the function $f()$ transforms the token according to the system specifications, that is, $f()$ is defined by the system and included in $pp$.
\item  $AGAT.Verify(pp,\; \mathcal{T},\;Priv_V)$ is run by $V$, and outputs 1 if $\mathcal{T}$ is a valid token and 0 otherwise. Note that $AGAT.Verify$  internally
checks the policy $NToken_h \leq k$ (rate limit). 
\item $AGAT.Open(pp,\;\mathcal{T},\;Priv_{TA})$ is run by $TA$ and takes $pp$, the token $\mathcal{T}$, and the private parameters of the token authority $Priv_{TA}$, and outputs the actual identity $H$ of the issuer. 
\end{itemize}

\textbf{Security requirements}
 A $k$-$AGAT$ scheme has the following properties: correctness, unforgeability, anonymity, and traceability.

\begin{itemize}
    \item \textbf{Correctness} ensures that AGAT token generation is correct, if  for an honest host and guest, $AGAT.Issue$ and \\ $AGAT.Show$ create a token $\mathcal{T}$  that $AGAT.Verify$ accepts, 
    even if the adversary can schedule joining of the hosts and guest and choose their identities and messages.
        \item \textbf{Unforgeability} ensures that if a guest token $\mathcal{T}$ is generated, the token generation has actually occurred. Additionally, it captures the fact that the adversary cannot generate more than $k$ valid tokens per verifier $V$ (the rate limit is done as part of $AGAT.Verify$ algorithm).
   More formally, we require that for all PPT adversaries $\mathcal{A}$ there exists a negligible function $\nu(\lambda)$ such that the adversary who can add corrupted members or corrupt the honest ones after enrollment and is given access to $AGAT.Issue$ 
    oracle should have only a negligible probability of outputting a forged token. 
    \item \textbf{Anonymity} ensures that the host, guest, and host-guest relation  cannot be learned by the verifier. 
    We define this property as the unlinkability of the host and guests identities, that is,
    given two guest identities $G_0$ and $G_1$ and two host identities $H_0$ and $H_1$, 
    the probability of linking the guests to hosts is negligible after seeing tokens issued by $AGAT.Issue$, $AGAT.Show$. We allow the adversary to make oracle calls to $AGAT.Issue$ and $AGAT.Show$ for different guests and hosts excluding the challenged ones.  

Note that this definition is strong and implies both the guest anonymity and host anonymity since an adversary who can distinguish either $G_0$ from $G_1$, or $H_0$ from $H_1$ by seeing the issued and presented tokens, can also distinguish the identity of the host and guest from the challenged token, and find the the guest-host relation.  
We note that our definition does not capture host and guest unlinkability, since we do not allow the adversary to see different tokens issued  or presented by the challenged entities other than the ones returned by $AnonChal_b$. Excluding these queries on challenged identities prevent trivial attacks where the adversary adds a host or a guest and sees their tokens by querying $AGAT.Issue$ and $AGAT.Show$, and later links these tokens to the tokens received from $AnonChal_b$. 
  \item \textbf{Traceability} protects the system by allowing 
  the token authority $TA$ to reveal the actual identity of the host $H$ who issued the token.
More formally, the probability that a PPT adversary $\mathcal{A}$
outputs a token $\mathcal{T}$ such that the $Open$ algorithm can not identify the host issuer is negligible. 
\end{itemize}

\textbf{Note.} The anonymity definition does not imply host and guest unlinkability with respect the verifier. This is intentional since (i) the tokens issued by the same host should be linked together in order to apply the rate limit policy on the number of tokens originated from the same host, and (ii) the guest is not trusted and we want to allow the service provider (verifier) to trace and analyze its access patterns. The anonymity property just ensures that the identity of the honest entities cannot be learned from their tokens and hence the relation between the actual host and guest remains anonymous to the verifier.

\textbf{Oracles.} To formally define the security requirements, we use the following oracles:

\begin{itemize}
    \item $\mathbf{AddH(H)}$: This oracle allows the adversary to add honest hosts through honest execution of $AGAT.Enroll$ without learning their private parameters.
    \item $\mathbf{AddC(H)}$: This oracle allows the adversary to add corrupted hosts. The adversary can deviate from the $AGAT.Enroll$ protocol and send arbitrary messages to $TA$ and see its output $out$. 
    \item $\mathbf{CorrU(id)}$ allows the adversary to corrupt the host or guest and learn  their communication transcript and their private parameters (including private keys). 
    \item $\mathbf{Issue(m, G, V, H)}$ This oracle is used by the adversary to obtain a guest token for guest $G$ which can be verified by $V$, from an honest host $H$ whose private parameters $Priv_h$ are not known by the adversary. It returns a token $\mathcal{T}$ on the message $m$ using the $AGAT.Issue(pp, G, m, Pub_V,Priv_h)$ algorithm.
    \item $\mathbf{Show(\mathcal{T}, G)}$ This oracle is used by the adversary to obtain a token that can be presented by $G$. It returns a token $\mathcal{T}'=f(\mathcal{T})$ on the message $m$ using the $AGAT.Show(pp, \mathcal{T}, Priv_g)$ algorithm.
    \item $\mathbf{Open(\mathcal{T})}$ returns the identity of the host issuer of $\mathcal{T}$.  This oracle cannot be called on a token obtained from the $AnonChal_b$ oracle.
    \item $\mathbf{AnonChal_b(pp, m, H_0,H_1, G_0, G_1)}$: This is a left-right oracle for defining anonymity. It takes as input the group public parameters $pp$, a message $m$, two host identities, $H_0$, $H_1$, and two guest identities, $G_0$ and $G_1$, and returns a token $\mathcal{T}_b$ issued by $H_b$ and presented by $G_b$ for a randomly chosen $b \in \{0,\;1\}$ 
    on the message $m$. The adversary can call this oracle once.
\end{itemize}

\begin{definition}
    For any security parameter $\lambda \in \mathrm{N}$ and for any PPT adversary $\mathcal{A}$, we say that $k$-$AGAT$ provides:
    \begin{enumerate}
        \item \textbf{Correctness} if there exists a negligible function $\nu_1$ such that 
        $Adv^{Corr}_{k\text{-}AGAT,\mathcal{A}} (\lambda)=Pr[Exp^{Corr}_{k\text{-}AGAT, \mathcal{A}}(\lambda)=1] \geq 1-\nu_1(\lambda)$\\

        \item \textbf{Unforgeability} if there exists a negligible function $\nu_2$ such that 
        $Adv^{Unforge}_{k\text{-}AGAT,\mathcal{A}} (\lambda)=Pr[Exp^{Unforge}_{k\text{-}AGAT, \mathcal{A}}(\lambda)=1] \leq \nu_2(\lambda)$\\

        \item \textbf{Anonymity} if there exists a negligible function $\nu_3$ such that  
        $Adv^{Anon}_{k\text{-}AGAT,\mathcal{A}} (\lambda)=Pr[Exp^{Anon}_{k\text{-}AGAT, \mathcal{A}}(\lambda)=1] \leq \frac{1}{2}+\nu_3(\lambda)$\\

        \item \textbf{Traceability} if there exists a negligible function $\nu_4$ such that 
        $Adv^{Trace}_{k\text{-}AGAT,\mathcal{A}} (\lambda)=Pr[Exp^{Trace}_{k\text{-}AGAT, \mathcal{A}}(\lambda)=1] \leq \nu_4(\lambda)$
    \end{enumerate}

    where $Exp^{Corr}_{k\text{-}AGAT, \mathcal{A}}$, $Exp^{Unforge}_{k\text{-}AGAT, \mathcal{A}}$, $Exp^{Anon}_{k\text{-}AGAT, \mathcal{A}}$, $Exp^{Trace}_{k\text{-}AGAT, \mathcal{A}}$ are defined in Figure \ref{AGATexperiments}.\\
\end{definition}

\begin{figure}[!bht]

\fbox{
    \begin{minipage}{0.95\columnwidth}
    \begin{small}
    $Exp^{Corr}_{k\text{-}AGAT, \mathcal{A}}$
    \hrule 
    $(pp,Priv_{TA}) \leftarrow AGAT.Setup(\lambda)$, $N=0$\\
    $(m, H, G, V) \leftarrow \mathcal{A}^{AddH,CorrU,AddC}(pp)$\\
    If $H \notin \mathcal{H}$  or $G \notin \mathcal{H}$ Return $0$\\
    $\mathcal{T} \leftarrow AGAT.Issue(pp, G, m, Pub_V, Priv_h)$\\
    $\mathcal{T}' \leftarrow AGAT.Show(pp, \mathcal{T},Priv_g) param_{Sm})$\\
    If $ AGAT.Verify(pp, \mathcal{T}', Priv_V)=0$ Return $0$\\
    \textbf{Return} 1\\

    $Exp^{Unforge}_{k\text{-}AGAT, \mathcal{A}}$
    \hrule 
    $(pp,Priv_{TA}) \leftarrow AGAT.Setup(\lambda)$, $N=0$, $Q_{Issue}=\emptyset$\\
    $(m, H, G, \mathcal{T}) \leftarrow \mathcal{A}^{AddH,CorrU,AddC,Issue}(pp)$\\
     If $H \notin \mathcal{H}$  or $G \notin \mathcal{H}$  Return $0$\\
    $\mathcal{T}' \leftarrow AGAT.Show(pp,\mathcal{T},Priv_g)$\\
    If $AGAT.Verify(pp, \mathcal{T}', Priv_V)=0$ Return $0$\\
    If $(m, \mathcal{T}) \notin Q_{Issue}$ Return $1$\\

    $Exp^{Anon-b}_{k\text{-}AGAT, \mathcal{A}}$
    \hrule 
    $(pp,Priv_{TA}) \leftarrow AGAT.Setup(\lambda)$, $N=0$, $Q_{AChal}=\emptyset$\\
     $b' \leftarrow \mathcal{A}^{AddH,CorrU,AddC,Issue,Show,Open,AnonChal_b}(pp)$\\
     If $b' \neq b$ Return $0$\\
    \textbf{Return} 1\\ 

     $Exp^{Trace}_{k\text{-}AGAT, \mathcal{A}}$
    \hrule 
    $(pp,Priv_{TA}) \leftarrow AGAT.Setup(\lambda)$, $N=0$\\
    $(m,\mathcal{T}')  \leftarrow \mathcal{A}^{AddH,CorrU,AddC,Verify, Open}(pp)$\\
     If $AGAT.Verify(pp, \mathcal{T}', Priv_V)=0$ Return $0$\\
    \textbf{Return} $AGAT.Open(pp, \mathcal{T}', Priv_{TA}) \overset{?}= \perp$\\
    
    \end{small}
    \end{minipage}
    }
     \caption{Security games of $k\text{-}AGAT$}
     \label{AGATexperiments}
\end{figure}

\vspace{1em}
\textbf{Our construction} is generic and
 uses a sponsored group signature $SPGS$,
 and a public key encryption scheme $E$ as its building blocks (please see Section \ref{Prelim} and \ref{SPGS} for their algorithms and security properties). 
 
In the nutshell, in our scheme, host act as the sponsor and the guest acts as the sponsored-member of a SPGS scheme. Host uses a sponsor signature to issue a guest token, and the guest uses a sponsored-member signature when it wants to show the token to the verifier. The important point in our scheme is that, 
 SPGS provides complete anonymity that captures unlinkability for sponsors, but the token generation system should not ensure complete anonymity. To enforce the rate limit policy, the token generation should allow host linkability. For this, we use the idea of using pseudonyms together with a group signature to relax their anonymity which is inspired by existing works such as \cite{belenkiy2009randomizable}. In our scheme, pseudonyms are generated by $TA$ and shared with hosts during their enrollment. $TA$ maintains a pseudonym list $L_w$ that consists of dummy pseudonyms initially; when the hosts join, $TA$ adds their pseudonyms to the list in batches and after shuffling them it shares the list with the verifier. Verifier do not see the linkage between pseudonyms and identity of the hosts and only use the list $L_w$ to check the validity of the pseudonyms. In guest token issuing algorithm, the host encrypts their pseudonym $w$ using the public key of the verifier and include it in the guest token together with the SPGS sponsor signature. The verifier decrypts the ciphertext and checks whether $w$ is among the valid pseudonyms $w \in L_w$ and if it is valid it keeps a counter for $w$ to count the number of issued tokens by the same host. Although one can use a different approach and let the hosts commit to their pseudonyms and prove their validity to the verifier through the zero knowledge proofs\footnote{This approach is publicly verifiable and one can use it to relax the assumption about the verifier's trustworthiness.} without sharing the list $L_w$ with the verifier \cite{belenkiy2009randomizable}, we chose to share $L_w$ with the verifier directly for checking the validity of pseudonyms to keep our construction simple and efficient. We show that as long as the verifier is trusted this approach ensures a secure $k$-$AGAT$ construction with the tradeoff that the pseudonyms should be added in batches and this can create a delay between when the host enrolls and when it can issue valid tokens. Additionally, to prevent the guests to steal the pseudonyms, and mix and match of the ciphertext and the the SPGS sponsor signatures in the token, we let the host to choose a fresh random value $r$ and include it both in the ciphertext and the SPGS sponsor signature every time it generates a new guest token.  Please see the details of construction below:

\begin{enumerate}
    \item  \underline{$AGAT.Setup(1^\lambda)$} is run by $TA$ and outputs the system parameters $pp$ (related to  $SPGS$ and $E$). $TA$ sets the set up parameters $setpp$ for the SPGS, and runs $(pp', msk) \leftarrow SPGS.Setup(1^\lambda,setpp)$, and obtains $(pp',msk)$. The $TA$'s private key is set to be the SPGS master private key $Priv_{TA}=msk$. $TA$ also receives the public key $Pub_V$ of the valid verifiers $V$. Also, it initializes the list of pseudonyms $L_w$ with some random dummy values. $TA$ outputs $pp=pp'||Pub_{V}||L_w$. 
\item \underline{$AGAT.Enroll(pp,\;H,\;Priv_{TA})$} is run by $TA$ and the host $H$ when the host joins. It runs $param_{Sp} \leftarrow Join^{Sp}(H, Priv_{TA})$, and chooses a random pseudonym $w_h$ for the host and outputs $Priv_h=param_{Sp}||w_h$. It also adds $w_h$ to $L_w$.
\item \underline{$AGAT.Issue(pp,\;G,\; m,\;Pub_V,\;Priv_h)$} is run between host $H$ and guest $G$. It first runs $param_{Sm} \leftarrow Join^{Sm}(G,Priv_h)$. Then, it chooses a fresh randomness $r$ and runs $\sigma \leftarrow  Sign^{Sp}(m||r, Priv_h)$, and $C= E.Enc(Pub_V,w_h||r)$. It outputs a guest token $\mathcal{T}= (m, r, \sigma, C) $ and $Priv_g=param_{Sm}$ to the guest. 
\item \underline{$AGAT.Show(pp, \; \mathcal{T}, \;Priv_g)$} is run by $G$. It parse the token as $\mathcal{T}= (m, r, \sigma, C) $, runs $\sigma' \leftarrow Sign^{Sm}(m||r, Priv_g)$ and outputs a new token $\mathcal{T'}$ which is $\mathcal{T'}=(\mathcal{T}, \sigma')$.
\item  \underline{$AGAT.Verify(pp,\;\mathcal{T}',\;Priv_V)$} is run by $V$; it parses the guest token as $\mathcal{T'}=(m, r, \sigma, C, \sigma')$. It decrypts $C$, and obtains $(w'||r')=E.Dec(Priv_V,\;C)$ and checks whether $r'=r$. If so, it stores $(w',\;r',\;\mathcal{T}')$. It then outputs 1 if (i) $SPGS.Verify^{Sp}(m||r, \sigma,pp)$ outputs 1,  (ii) $SPGS.Verify^{Sm}(m||r, \sigma',pp)$ outputs 1, (iii) $w'$ is in the list $L_w$,  and $w'$ has not issued more than $k$ tokens. It also increases the counter $NToken[w']$ by one.
Otherwise, outputs 0.
\item \underline{$AGAT.Open(pp,\; \mathcal{T},\; Priv_{TA})$} is run by $TA$ which takes $Priv_{TA}$, parses the token as $\mathcal{T'}=(m, r, \sigma, C, \sigma')$,  and outputs $H'=SPGS.Open^{Sm}(Priv_{TA},\sigma')=\\SPGS.Open^{Sp}(Priv_{TA},\sigma)$.
\end{enumerate}

\textbf{Concrete construction.}
To instantiate our construction, we use the SPGS  concrete construction of section \ref{SPGSInit}. 
For the encryption scheme $E$, we  use the EC-based Elgamal encryption scheme. 

\textbf{Security analysis.}
Below we give a theorem and a proof sketch for our construction of $k$-$AGAT$.
\begin{theorem}
    Our $k$-$AGAT$ construction ensures correctness, unforgeability, $\frac{1}{b_n}$-anonymity, and traceability assuming $SPGS$ satisfies correctness, non-frameability, sponsor anonymity, sponsored-member privacy, sponsor traceability, and $E$ is an IND-CPA secure encryption scheme.
\end{theorem} 

\textit{Proof sketch.} We omit the proof due to the space and just give the informal arguments. \\
\textit{Correctness} is satisfied due to the  correctness guarantee of the sponsored group signature scheme $SPGS$ and the public key encryption scheme $E$.\\
\textit{Unforgeability} follows from (i) the non-frameability of the sponsored group signature scheme which prevent the adversary to generate a valid sponsor signature $\sigma$ for a guest token  without knowing the private key of the host, (ii) the randomness $r$ used in both signature $\sigma$ and ciphertext $C$ prevents the adversary to mix and match different signatures and ciphertexts from different tokens to form a new token. \\
\textit{Anonymity} follows from (i) anonymity of the sponsored group signature scheme which does not let the verifier learn whether the sponsor signature $\sigma$ in the token has been generated by $H_0$ or $H_1$, (ii) sponsored-member privacy of SPGS which prevents the verifier to learn the identity of the guest from the sponsored-member signature used in the token, (iii) the IND-CPA security  of $E$ which does let the verifier to learn the pseudonym of the host, or even distinguish the encrypted pseudonym from random, (iv) the fact that the  registered pseudonyms are independent of the real identity of the hosts, (v) pseudonyms are added to $L_w$ in batches of size $b_n$ and the verifier who controls the join of hosts and sees the latest list $L_w$ cannot guess the pseudonym of an honest host with probability greater than $\frac{1}{b_n}$. \\
\textit{Issuer traceability} is ensured due to the traceability of the sponsored group signature scheme.

\end{document}